\documentclass[reqno]{amsart}
\usepackage{amsmath}
\usepackage{amsfonts}

\usepackage{amsmath,amssymb}
\usepackage{graphicx}
\usepackage{color}

\newtheorem{theorem}{Theorem}
\newtheorem{lemma}{Lemma}
\newtheorem{corollary}{Corollary}

\newtheorem*{conclusion*}{Conclusion}
\newtheorem*{definition*}{Definition}
\newtheorem*{remark*}{Remark}

\newtheorem{example}{Example}

\makeatletter

\newcommand{\Rmnum}[1]{\expandafter\@slowromancap\romannumeral #1@}
\makeatother

\newcommand{\ls}[1]
    {\dimen0=\fontdimen6\the\font\lineskip=#1\dimen0
     \advance\lineskip.5\fontdimen5\the\font
     \advance\lineskip-\dimen0
     \lineskiplimit=0.9\lineskip
     \baselineskip=\lineskip
     \advance\baselineskip\dimen0
     \normallineskip\lineskip\normallineskiplimit\lineskiplimit
     \normalbaselineskip\baselineskip
     \ignorespaces}

\begin{document}

\bibliographystyle{abbrv}

\title{Cyclic codes from  the second class two-prime Whiteman's generalized cyclotomic sequence with order 6}
\author{Pramod Kumar Kewat and Priti Kumari}
\address{Department of Applied Mathematics, Indian School of Mines, Dhanbad 826 004, India}
\email{kewat.pk.am@ismdhanbad.ac.in, priti.jsr13@gmail.com}        
\subjclass{}
\keywords{Cyclic codes, finite fields, cyclotomic sequences}
\date{}
 \maketitle
\markboth{P.K. Kewat and Priti Kumari}{Cyclic codes from  the second class two-prime WGCS with order 6}

\thispagestyle{plain} \setcounter{page}{1}

\begin{abstract}
Let $n_1=ef+1$ and  $n_2=ef'+1$ be two distinct odd primes with positive integers $e,\ f,\ f'.$ In this paper, the two-prime Whiteman's generalized cyclotomic sequence of order $e=6$
is employed to construct several classes of cyclic codes over $\mathrm{GF}(q)$ with length $n_1n_2$. The lower bounds on the minimum distance
of these cyclic codes are obtained.\\

\end{abstract}

\ls{1.5}
\section{Introduction}\label{section 1}
Let $q$ be a power of a prime $p$. An $[n,k,d]$ linear code $C$ over a 
finite field $\mathrm{GF}(q)$ is a $k-$dimensional subspace of the vector space $\mathrm{GF}(q)^n$ with the minimum distance $d$.
A linear code $C$ is a cyclic code if the cyclic shift of a codeword in $C$ is again a codeword in $C$, i.e., 
if $ (c_0,\cdots,c_{n-1})\in C$ then $(c_{n-1},c_0\cdots,c_{n-2})\in C$.
Let $\mathrm{gcd}(n, q)= 1$. We consider the univarivate polynomial ring $\mathrm{GF}(q)[x]$ and the ideal 
$I=\langle x^n-1\rangle$ of $\mathrm{GF}(q)[x].$ We denote by $R$ the ring  $\mathrm{GF}(q)[x]/I$.
We can consider a cyclic code of length $n$ over $\mathrm{GF}(q)$ as an ideal in $R$ via the following correspondence
\begin{align}
 \mathrm{GF}(q)^n \rightarrow  R,\ \ \ \ \ (c_0,c_1,\cdots,c_{n-1}) \mapsto c_0+c_1x+\cdots+c_{n-1}x^{n-1}.\nonumber
\end{align}
 Then, a linear code $C$ over $\mathrm{GF}(q)$ is a cyclic code if and only if $C$ is an ideal of $R$. Since $R$ is a principal ideal ring, if $C$ is not trivial, there exists a unique monic polynomial
 $g(x)$ dividing $x^n-1$ in $\mathrm{GF}(q)[x]$ and $C=\langle g(x) \rangle$. The polynomials $g(x)$ and $h(x)=(x^n-1)/g(x)$ are called
 the generator polynomial and the parity-check polynomial of $C$ respectively.
 If the dimension of the code $C$ is $k$, the generator polynomial has degree $n-k.$  An $[n,k,d]$ cyclic code $C$ is capable of encoding $q-$ary messages of length $k$ and requires $n-k$ 
 redundancy symbols.
 
 The total number of cyclic codes over $\mathrm{GF}(q)$ and their construction are closely related to the cyclotomic cosets modulo $n$. 
 One way to construct cyclic codes over $\mathrm{GF}(q)$ with length $n$ is to use the generator polynomial
\begin{align}
 \frac{x^n-1}{\mathrm{gcd}(x^n-1,S(x))}\label{genpol},
\end{align}
where $S(x)=\sum\limits_{i=0}^{n-1}s_ix^i \in \mathrm{GF}(q)[x]$ and $s^\infty={(s_i)}_{i=0}^{\infty}$ is a sequence of period $n$ over $\mathrm{GF}(q).$
The cyclic code $C_s$ generated by the polynomial in Eq. (\ref{genpol}) is called the cyclic code defined by the sequence $s^{\infty},$ and the sequence $s^{\infty}$ is called the defining sequence of the cyclic code $C_s.$ 

Cyclic codes have been studied in a series of papers and a lot of progress have been accomplished 
 (see, for example \cite{Betti06}, \cite{Ding15}, \cite{Eupen93},  \cite{Macwilliams77} and \cite{Lint86}).
 The Whiteman generalized cyclotomy was introduced by Whiteman and its properties were studied in \cite{whiteman62}.
 The two-prime Whiteman's generalized cyclotomic sequence(WGCS) was introduced by Ding \cite{Ding97} and its coding properties were studied in \cite{ding12} and \cite{white04}.
 Ding \cite{ding12} and Sun et al.\cite{white04} constructed number of classes of
cyclic codes over $\mathrm{GF}(q)$ with length $n=n_1n_2$ from the two-prime Whiteman's generalized cyclotomic sequences of order $2$ and $4$ respectively and gave the lower bounds
on the minimum weight of these cyclic codes under certain conditions. Li et al. \cite{Linearcom} gave a lower bound on linear complexity of WGCS of order six, which indicates
the linear complexity is large. The autocorrelation values of WGCS were determined in \cite{Auto13}.
Inspired by the construction of above two papers ( \cite{ding12} and \cite{white04}), in this paper, we employ the two-prime
Whiteman's generalized cyclotomic sequences with order $6$ to construct several classes of cyclic codes over $\mathrm{GF}(q)[x].$  We also obtain lower bounds 
on the minimum weight of these cyclic codes.
 \section{Preliminaries}\label{section 2}
 
\subsection{Linear complexity and minimal polynomial}\label{subsection 2.1}

If ${(s_i)}_{i=0}^{\infty}$ is a sequence over a finite field $\mathrm{GF}(q)$ and $f(x)$ is a polynomial with coefficients in $\mathrm{GF}(q)$ given by
\begin{align}
 f(x)=c_0+c_1x+\cdots+c_{L-1}x^{L-1},\nonumber
\end{align}
then we define 
\begin{align}
 f(E)s_j=c_0s_j+c_1s_{j-1}+\cdots+c_{L-1}s_{j-L+1},\nonumber
\end{align}
where $E$ is a left shift operator defined by $Es_i=s_{i-1}$ for $i\geq 1.$
Let $s^n$ be a sequence $s_0s_1\cdots s_{n-1}$ of length $n$ over a finite field $\mathrm{GF}(q)$. For a finite sequence, the $n$ is finite; for a semi-infinite sequence, the $n$ is $\infty$.
A polynomial $f(x)\in \mathrm{GF}(q)[x]$ of degree $\leqslant l$ with $c_0\neq 0$ is called a characteristic polynomial of the sequence $s^n$ if $f(E)s_j=0$ for all $j$ with $j\geq l.$
For every characteristic polynomial there is a least $l\geq \mathrm{deg}(f)$ such that the above equation hold. The smallest $l$  is called the associate recurrence length of $f(x)$
with respect to the sequence $s^n$. The characteristic polynomial with smallest length is known as minimal polynomial of the sequence $s^n$ and the associated recurrence 
length is called the linear span or linear complexity of the sequence $s^n$.\\ If a semi-infinite sequence $s^{\infty}$ is periodic, then its minimal polynomial is unique 
if $c_0=1.$ The linear complexity of a periodic sequence is equal to the degree of its minimal polynomial. For the periodic sequences,
there are few ways to determine their linear spans and minimal polynomials. One of them is given in the following lemma.
\begin{lemma}\label{spanmin}
\cite{finite} Let $s^\infty$ be a sequence of a period $n$ over $GF(q)$. Define
 \begin{align}
 S^{n}(x)=s_0+s_{1}x +\cdots+s_{n-1}x^{n-1}\in \mathrm{GF} (q) [x].\nonumber
 \end{align}
 Then the minimal polynomial $m_s$ of $s^\infty$ is given by
 \begin{align}
  \frac{x^n-1}{\mathrm{gcd}(x^n-1,S^n(x))}.\label{minpol}
 \end{align}
 Consequently, the linear span $L_s$ of $s^\infty$ is given by
 \begin{align}
  L_s=n-deg(\mathrm{gcd}(x^n-1,S^n(x))).\label{linearspan}
 \end{align}
\end{lemma}

\subsection{The Whiteman's generalized cyclotomic sequences and its construction}\label{subsection 2.2}
An integer $a$ is called a primitive root of modulo $n$ if the multiplicative order of $a$ modulo $n,$ denoted by $ \mathrm{ord}_n(a)$, is equal to $\phi(n),$ where $\phi(n)$ is the Euler phi function and $\mathrm{gcd}(a,n)=1.$ \\
Let $n_1$ and $n_2$ be two distinct odd primes, define $n=n_1n_2,\ d=\mathrm{gcd}(n_1-1,n_2-1)$ and $e=(n_1-1)(n_2-1)/d$. From the Chinese Remainder theorem
, there are common primitive roots of both $n_1$ and $n_2$. Let $g$ be a fixed common primitive root of both $n_1$ and $n_2$. Let $u$ be an integer satisfying
\begin{align}
 u\equiv g \ (\mathrm{mod}\ n_1),~~~~~  u\equiv 1 \ (\mathrm{mod}\ n_2).\label{defu}
\end{align}
Whiteman \cite{whiteman62} proved that
\begin{align}
 \mathbb{Z}_{n}^{\ast}=\{g^{s}u^{i}:s=0,1,\cdots,e-1;\ i=0,1,2,\cdots,d-1\}.\nonumber
\end{align}
where $\mathbb{Z}_{n}^{\ast}$ denotes the set of all invertible elements of the residue class ring $\mathbb{Z}_{n}$ and $e$ is the order of $g$ modulo $n$.
The Whiteman's generalized cyclotomic classes $W_i$  of order $d$ are defined by
\begin{align}
 W_i=\{g^su^i ~(\mathrm{mod}\ n) :s=0,1,\cdots,e-1\},i=0,1,\cdots,d-1.\nonumber
\end{align}
The classes $W_i$, $1\leq i\leq d-1$ give a partition of $ \mathbb{Z}_{n}^{\ast},$ i.e.,
 $\mathbb{Z}_{n}^{\ast}=\cup_{i=0}^{d-1} {W_i},\ W_{i}\cap W_{j}=\emptyset \ \mathrm{for}\ i\neq j.$ \\
Let
$$P=\{n_{1},2n_{1},3n_{1},\cdots,(n_{2}-1)n_{1}\},\ Q=\{n_{2},2n_{2},3n_{2},\cdots,(n_{1}-1)n_{2}\},$$
$$C_{0}=\{0\}\cup Q\cup \bigcup\limits_{i=0}^{\frac{d}{2}-1}W_{i},\ C_{1}=P\cup \bigcup\limits_{i=\frac{d}{2}}^{d-1}W_{i},$$
$$C_{0}^{\ast}=\{0\}\cup Q\cup \bigcup\limits_{i=0}^{\frac{d}{2}-1}W_{2i}\ and \ \ C_{1}^{\ast}=P\cup \bigcup\limits_{i=0}^{\frac{d}{2}-1}W_{2i+1}.$$
It is easy to see that if $d > 2$, then $C_0 \neq C_0^{\ast}$ and $C_1 \neq C_1^{\ast}$ . 
Now, we introduce two kinds of Whiteman's generalized cyclotomic sequences of order $d$ (see \cite{defwhiteman}).\\
\begin{definition*} The two-prime Whiteman's generalized cyclotomic sequence $(s^*)^\infty= (s_i^*)_{i=0}^{n-1} $ of order $d$ and period $n$, 
which is called two-prime WGCS-I, is defined by
\begin{equation}
s_{i}^*=\left\{
\begin{array}{ll}
0,\ \ \mathrm{if}\ i\in C^*_{0},\\
1,\ \ \mathrm{if}\ i\in C^*_{1}.\nonumber
\end{array}
\right.
\end{equation}
The two-prime Whiteman's generalized cyclotomic sequence $s^\infty=(s_i)_{i=0}^{n-1}$ of order $d$ and period $n$, 
which is called two-prime WGCS-II, is defined by
\begin{equation}
s_{i}=\left\{
\begin{array}{ll}
0,\ \ \mathrm{if}\ i\in C_{0},\\
1,\ \ \mathrm{if}\ i\in C_{1}.\label{si}
\end{array}
\right.
\end{equation}
\end{definition*}
The cyclotomic numbers corresponding to these cyclotomic classes are defined as
$$
(i,j)_{d}=|(W_{i}+1)\cap W_{j}|, \nonumber where~~ 0\leq i,j\leq d-1.
$$
Additionally, for any $t\in \mathbb{Z}_{n}$, we define
\begin{align}
d(i,j;t)=|(W_{i}+t)\cap W_{j}|,\nonumber
\end{align}
where $W_{i}+t=\{w+t|w\in W_{i}\}$.\\

\section{A class of cyclic codes over $\mathrm{GF}(q)$ defined by the two-prime WGCS}\label{section 3}

\subsection{Properties of the Whiteman's cyclotomy of order d}\label{section 3.1}
In this subsection, we summarize number of properties of the Whiteman's generalized cyclotomy of order $d=\mathrm{gcd}(n_1-1,n_2-1).$
The following Lemma follows from the Theorem 4.4.6 of \cite{Streamcipher}.
\begin{lemma}\label{value d(i,j,t)}

Let the notations be same as before and $t\neq0$. We have
\begin{align}
d(i,j;t)=\left\{
\begin{array}{llll}
\frac{(n_{1}-1)(n_{2}-1)}{d^2},\ \ \ \ \ \ \ \ \ \ \ \ \ \ \ \ \  i\neq j,\ t \in  P \cup Q,  \nonumber\\
\frac{(n_{1}-1)(n_{2}-1-d)}{d^2},\ \ \ \ \ \ \ \ \ \ \ \ \ \ i=j,\ t\in P,\ t \notin Q,\nonumber \\
\frac{(n_{1}-1-d)(n_{2}-1)}{d^2},\ \ \ \ \ \ \ \ \ \ \ \ \ \ i=j,\ t\in Q,\ t \notin P,\nonumber\\
(i',j')_d \ for\  some\ i',j',\ \ \ \ \ otherwise.
\end{array}
\right.
\end{align}
\end{lemma}

\begin{lemma}\label{even}
Let the symbols be defined as before. The following four statements are equivalent:\\
\rm{(1)} $-1 \in W_{\frac{d}{2}}.$\\
\rm{(2)}  $\frac{(n_1-1)(n_2-1)}{d^2}$ is even.\\
\rm{(3)} One of the following sets of equations are satisfied:\\
 $
 \left\{
\begin{array}{ll}
 n_1 \equiv 1 \ (\mathrm{mod}\ 2d)\nonumber\\
 n_2 \equiv d+1 \ (\mathrm{mod}\ 2d),\nonumber
\end{array}
\right.
 \left\{
\begin{array}{ll}
 n_1 \equiv d+1 \ (\mathrm{mod}\ 2d)\nonumber\\
 n_2 \equiv 1 \ (\mathrm{mod}\ 2d).\nonumber
\end{array}
\right.
$\\
\rm{(4)} $ n_1n_2 \equiv d+1 \ (\mathrm{mod}\ 2d).$
\end{lemma}
\begin{proof} $(1) \Leftrightarrow (2)$  The result follows from (2.3) in \cite{whiteman62}.\\
$(2) \Rightarrow (3) $  Let $n_1-1=df,\  n_2-1=df'$ and $e=dff'$, where $f$ and $f'$ are integer. Since $(f,f')=1,\ f$ and $f'$ can not both be even. Here $ff'=\frac{(n_1-1)(n_2-1)}{d^2}$ is even. So, $f$ or $f'$ is even.
Let $f$ is even and $f'$ is odd. If $f$ is even, then $n_1-1 = d(2k_1)$, where $k_1$ is an integer. Therefore, $n_1 \equiv 1 \ ( \mathrm{mod}\ 2d )$.  If $f'$  is odd, then $n_1-1 = d(2k_2+1)$, where $k_2$ is an integer. Therefore, $n_2 \equiv d+1 \ ( \mathrm{mod}\ 2d ).$ 
Similarly, when $f$ is odd and $f'$ is even. We get $n_1 \equiv d+1 \ ( \mathrm{mod}\ 2d )$ and $n_2 \equiv 1 \ ( \mathrm{mod}\ 2d ).$\\
$(3) \Rightarrow (2)$  and $(3) \Rightarrow (4)$ are obvious.\\
$(4) \Rightarrow (3)$  Since, $\mathrm{gcd}(n_1-1,n_2-1)=d$, let $n_1-1=fd$ and $n_2-1=f'd$. We have $ n_1n_2 \equiv d+1\ (\mathrm{mod}\ 2d)$, this gives $fd+f'd\equiv d \ (\mathrm{mod}\ 2d).$
Thus, we have $f+f'=2k+1$ for an integer $k$. So, $n_1=2kd+(1-f')d+1$, this gives $n_1\equiv(1-f')d+1 \ (\mathrm{mod}\ 2d).$
If $f'$ is odd, then $n_1\equiv 1 \ (\mathrm{mod}\ 2d)$ and $n_2 \equiv d+1 \ (\mathrm{mod}\ 2d)$. If $f'$ is even, then $n_1\equiv d+1 \ (\mathrm{mod}\ 2d)$ and $n_2 \equiv 1 \ (\mathrm{mod}\ 2d).$ 
\end{proof}
\begin{lemma}\label{odd}
Let the symbols be defined as before. The following four statements are equivalent:\\
\rm{(1)} $-1 \in W_0.$\\
\rm{(2)}  $\frac{(n_1-1)(n_2-1)}{d^2}$ is odd.\\
\rm{(3)}  The following set of equation is satisfied:\\
 $
 \left\{
\begin{array}{ll}
n_1 \equiv d+1 \ (\mathrm{mod}\ 2d)\nonumber\\
n_2 \equiv d+1 \ (\mathrm{mod}\ 2d),\nonumber
\end{array}
\right.
$\\
\rm{(4)} $ n_1n_2 \equiv1 \ (\mathrm{mod}\ 2d).$
\end{lemma}
\begin{proof} 
Similar to the proof of the above lemma.
\end{proof}

\subsection{Properties of Whiteman's cyclotomy of order 6}\label{section 3.2}
We recall the following lemmas (Lemma 1 and Lemma 2) from \cite{gong13}. 
\begin{lemma}\label{odd&even}
Let $\mathrm{gcd}(n_1-1,n_2-1) = 6$, i.e.,  $n_1 \equiv 1\ \mathrm{mod\ 6}$, $n_2 \equiv 1\ \mathrm{mod\ 6}$. 
Let $a,b,x,y,c$ and $d$ are integers. There are 10 possible different cyclotomic numbers of order 6 and they are given by the following relations:
\\ If $\frac{(n_{1}-1)(n_{2}-1)}{36}$ is odd, we have
 $$
 \begin{array}{lllll}
 &(0,0)_{6}=\frac{1}{72}(12M+32+6a-24x+2c),\\
 &(0,1)_{6}=(1,0)_{6}=(5,5)_{6}=\frac{1}{72}(12M+8+a+3b+8x+24y-c+9d),\\
 &(0,2)_{6}=(2,0)_{6}=(4,4)_{6}=\frac{1}{72}(12M+8-3a+9b-c-9d),\\
 &(0,3)_{6}=(3,0)_{6}=(3,3)_{6}=\frac{1}{72}(12M+8-2a+8x+2c),\\
 &(0,4)_{6}=(4,0)_{6}=(2,2)_{6}=\frac{1}{72}(12M+8-3a-c-9b+9d),\\
 &(0,5)_{6}=(5,0)_{6}=(1,1)_{6}=\frac{1}{72}(12M+8+a-3b+8x-24y-c-9d),\\
 &(1,2)_{6}=(2,1)_{6}=(4,5)_{6}=(5,4)_{6}=(5,1)_{6}=(1,5)_{6}=\frac{1}{72}(12M-4-2a-4x+2c),\\
 &(1,3)_{6}=(2,5)_{6}=(3,1)_{6}=(3,4)_{6}=(4,3)_{6}=(5,2)_{6}=\frac{1}{72}(12M-4+a+3b-4x-12y-c+9d),\\
 &(1,4)_{6}=(2,3)_{6}=(3,2)_{6}=(3,5)_{6}=(4,1)_{6}=(5,3)_{6}=\frac{1}{72}(12M-4+a-3b-4x+12y-c-9d) \ \mathrm{and} \\
 &(2,4)_{6}=(4,2)_{6}=\frac{1}{72}(12M-4+6a+12x+2c).
 \end{array}
 $$
 If $\frac{(n_{1}-1)(n_{2}-1)}{36}$ is even, we have
 $$
 \begin{array}{lllll}
 &(0,0)_{6}=(3,0)_{6}=(3,3)_{6}=\frac{1}{72}(12M+20-8x-2a+2c),\\
 &(0,1)_{6}=(2,5)_{6}=(4,3)_{6}=\frac{1}{72}(12M-4-3a-9b-c+9d),\\
 &(0,2)_{6}=(1,4)_{6}=(5,3)_{6}=\frac{1}{72}(12M-4-8x+a-c+24y-3b-9d),\\
 &(0,3)_{6}=\frac{1}{72}(12M-4+24x+6a+2c),\\
 &(0,4)_{6}=(1,3)_{6}=(5,2)_{6}=\frac{1}{72}(12M-4-8x+a-c-24y+3b+9d),\\
 &(0,5)_{6}=(2,3)_{6}=(4,1)_{6}=\frac{1}{72}(12M-4-3a-c+9b-9d),\\
 &(1,0)_{6}=(2,2)_{6}=(3,1)_{6}=(3,4)_{6}=(4,0)_{6}=(5,5)_{6}=\frac{1}{72}(12M+8+4x+a-c+12y+3b+9d),\\
 &(1,1)_{6}=(2,0)_{6}=(3,2)_{6}=(3,5)_{6}=(4,4)_{6}=(5,0)_{6}=\frac{1}{72}(12M+8+4x+a-c-12y-3b-9d),\\
 &(1,2)_{6}=(1,5)_{6}=(2,4)_{6}=(4,2)_{6}=(5,1)_{6}=(5,4)_{6}=\frac{1}{72}(12M-4+4x-2a+2c)\ \mathrm{and} \\
 &(2,1)_{6}=(4,5)_{6}=\frac{1}{72}(12M-4+6a-12x+2c).
 \end{array}
 $$
 Where $n_1n_2 = x^2+3y^2,\ M = \frac{1}{6}((n_1-2)(n_2-2)-1)$ and $4n_1n_2 = a^2+3b^2 = c^2+27d^2$. 
 \end{lemma}
 From [\cite{Streamcipher}, Theorem I.15 and page 118], we get following relation between the parameters $a,b,x,y,c$ and $d.$ 

 \begin{lemma}\label{relation}
  Let $k$ and $l$ be integer such that $k\equiv g \ \mathrm{mod}(n_1)$ and $ l\equiv g\ \mathrm{mod}(n_2).$ Suppose $k^\rho = 2 \in \mathbb{F}_{n_1}$ and $l^\varrho = 2 \in \mathbb{F}_{n_2}$ for some integers $\rho$ and $\varrho$.
  Let the parameters $a,b,x,y,c,d$ be same as in the above lemma. Then, we have the following results.\\
  Case-I If ${\{(n_{1}-1)(n_{2}-1)}\}/{36}$ is even, we have
  \begin{enumerate}\item[1)] If $\rho-\varrho\equiv 0 \ (\mathrm{mod}\ 3),$ then $a=2x=-c,\ b=-2y=-3d.$
   \item [2)] If $\rho-\varrho\equiv 1 \ (\mathrm{mod}\ 3),$ then $a=-x-3y,\ b=-x+y,\ c=x-3y,\ 3d=-x-y.$
   \item[3)] If $\rho-\varrho\equiv 2 \ (\mathrm{mod}\ 3),$ then $a=-x+3y,\ b=x+y,\ c=x+3y,\ 3d=x-y.$
  \end{enumerate}
  Case-II ${\{(n_{1}-1)(n_{2}-1)}\}/{36}$ is odd, we have
\begin{enumerate}\item[1)] If $\rho-\varrho\equiv 0 \ (\mathrm{mod}\ 3),$ then $a=-2x=c,\ b=2y=3d.$
   \item [2)] If $\rho-\varrho\equiv 1 \ (\mathrm{mod}\ 3),$ then $a=x+3y,\ b=x-y,\ c=x-3y,\ 3d=-x-y.$
   \item[3)] If $\rho-\varrho\equiv 2 \ (\mathrm{mod}\ 3),$ then $a=x-3y,\ b=-x-y,\ c=x+3y,\ 3d=x-y.$
  \end{enumerate}
\end{lemma}
 Furthermore, we get the following form for the cyclotomic numbers of order 6 after substituting the value of $a, b,c$ and $d$ from Lemma \ref{relation} to Lemma \ref{odd&even}.
 
 \begin{table}[htdp]\label{table1}
  \caption{The cyclotomic number of order 6 for even ${\{(n_{1}-1)(n_{2}-1)}\}/{36}$ }
  \centering
  \begin{center}
   \begin{tabular}{|c|c|c|c|}
    \hline
      &$\rho-\varrho\equiv0 \ (\mathrm{mod\ 3})$& $\rho-\varrho\equiv1 \ (\mathrm{mod\ 3})$&$\rho-\varrho\equiv2 \ (\mathrm{mod\ 3})$  \\ \hline
    $36(0,0)_6$ & 6M + 10 - 8x\ \ \ \ \ \ \ \ \ & 6M + 10 - 2x\ \ \ \ \ \ \ \ \ & 6M + 10 - 2x\ \ \ \ \ \ \ \ \ \ \\ \hline
    $36(0,1)_6$ & 6M - 2 - 2x + 12y\ \ \  & 6M - 2 + 4x\ \ \ \ \ \ \ \ \ \ \ & 6M - 2 - 2x - 12y\ \ \ \ \ \\ \hline
    $36(0,2)_6$ & 6M - 2 - 2x + 12y\ \ \  & 6M - 2 - 2x + 12y\ \ \ & 6M - 2 - 8x + 12y\ \ \ \ \\ \hline
    $36(0,3)_6$ & 6M - 2 + 16x\ \ \ \ \ \ \ \ \  & 6M - 2 + 10x - 12y & 6M - 2 + 10x + 12y\\ \hline
    $36(0,4)_6$ & 6M - 2 - 2x - 12y\ \ \ \ & 6M - 2 - 8x - 12y\ \ \ & 6M - 2 - 2x - 12y\ \ \ \ \ \\ \hline
    $36(0,5)_6$ & 6M - 2 - 2x - 12y\ \ \ \ & 6M - 2 - 2x + 12y & 6M - 2 + 4x\ \ \ \ \ \ \ \ \ \ \ \\ \hline
    $36(1,0)_6$ & 6M + 4 + 4x + 6y & 6M + 4 - 2x + 6y\  & 6M + 4 + 4x + 6y \ \\ \hline
    $36(1,1)_6$ & 6M + 4 + 4x - 6y\ \   & 6M + 4 + 4x - 6y & 6M + 4 - 2x - 6y\ \ \ \ \\ \hline
    $36(1,2)_6$ & 6M - 2 - 2x\  \ \ \ \ \ \ \ \ \ & 6M - 2 + 4x\ \ \ \ \ \ \ \ & 6M - 2 + 4x\ \ \ \ \ \ \ \ \ \ \\ \hline
    $36(2,1)_6$ & 6M - 2 - 2x\ \ \ \ \ \ \ \ \ \  & 6M - 2 - 8x - 12y \  & 6M - 2 - 8x + 12y\ \ \\ \hline
   \end{tabular}

  \end{center}
\label{default}
 \end{table}
 
\begin{table}[htdp]\label{tabel2}
  \caption{The cyclotomic number of order 6 for odd ${\{(n_{1}-1)(n_{2}-1)}\}/{36}$ }
  \centering
  \begin{center}
   \begin{tabular}{|c|c|c|c|}
    \hline
      &$\rho-\varrho\equiv0 \ (\mathrm{mod\ 3})$& $\rho-\varrho\equiv1 \ (\mathrm{mod\ 3})$&$\rho-\varrho\equiv2 \ (\mathrm{mod\ 3})$  \\ \hline
    $36(0,0)_6$ & 6M + 16 - 20x\ \ \ \ \ \ \ \ \ \ & 6M + 16 - 8x + 6y\ \ \  & 6M + 16 - 8x - 6y\ \ \ \ \ \ \ \\ \hline
    $36(0,1)_6$ & 6M + 4 + 4x + 18y\ \ \  & 6M + 4 + 4x + 12y  & 6M + 4 + 4x + 6y\ \ \ \ \ \\ \hline
    $36(0,2)_6$ & 6M + 4 + 4x + 6y\ \ \ \ \  & 6M + 4 + 4x - 6y\ \ \ \  & 6M + 4 - 8x \ \ \ \ \ \ \ \ \ \ \ \ \\ \hline
    $36(0,3)_6$ & 6M + 4 + 4x\ \ \ \ \ \ \ \ \ \ \ \  & 6M + 4 + 4x - 6y\ \ \ \ & 6M + 4 + 4x + 6y\ \ \ \ \\ \hline
    $36(0,4)_6$ & 6M + 4 + 4x - 6y\ \ \ \ \ \ & 6M + 4 - 8x \ \ \ \ \ \ \ \ \ \ \   & 6M + 4 + 4x + 6y\ \ \ \ \\ \hline
    $36(0,5)_6$ & 6M + 4 + 4x - 18y\ \ \ \ & 6M + 4 + 4x - 6y\ \ \ \  & 6M + 4 + 4x - 12y\ \ \ \ \\ \hline
    $36(1,2)_6$ & 6M - 2 - 2x\ \ \ \ \ \ \ \ \ \ \ \ \ \  & 6M - 2 - 2x - 6y\ \ \ \ \ \ \  & 6M - 2 - 2x + 6y\ \ \ \ \ \ \ \\ \hline
    $36(1,3)_6$ & 6M - 2 - 2x\ \ \ \ \ \ \ \ \ \ \ \ \ \   & 6M - 2 - 2x - 6y\ \ \ \ \ \ \ & 6M - 2 - 2x - 12y\ \ \ \ \ \ \ \\ \hline
    $36(1,4)_6$ & 6M - 2 - 2x\  \ \ \ \ \ \ \ \ \ \ \ \ \  & 6M - 2 - 2x + 12y\ \ \ \ & 6M - 2 - 2x + 6y\ \ \ \ \ \ \ \\ \hline
    $36(2,4)_6$ & 6M - 2 - 2x\ \ \ \ \ \ \ \ \ \ \ \ \ \  & 6M - 2 + 10x + 6y \  & 6M - 2 + 10x - 6y\ \ \ \ \ \\ \hline
   \end{tabular}

  \end{center}
\label{default}
 \end{table} 
These $36$ cyclotomic numbers $(i,j)$ are solely functions of the unique representation of $p=x^2+3y^2;\ x\equiv1 \ (\mathrm{mod}\ 3)$ and the sign of $y$ is ambiguously determined. 
 \begin{lemma}\label{-1} \cite{whiteman62}
Define $\eta=\frac{(n_{1}-1)(n_{2}-1)}{36}$. Let symbols be same as before. Then
$$
-1 \in \left\{
\begin{array}{llll}
W_0,\   \ \mathrm{if}\ \eta\   is\ odd,\\
W_3,\   \ \mathrm{if}\ \eta\   is\ even.
\end{array}
\right.
$$
\end{lemma}

\subsection{A class of cyclic codes over $\mathrm{GF}(q)$ defined by two-prime WGCS-II}\label{subsection 3.3} 

We have $\mathrm{gcd}(n,q)=1$. Let $m$ be the order of $q$ modulo $n$. Then the field $\mathrm{GF}(q^{m})$ has a primitive $nth$ root of unity $\beta$. We define
\begin{align}
S(x)=\sum\limits_{i\in C_{1}}x^{i}=\left(\sum\limits_{i\in P}+\sum\limits_{i\in W_{3}}+\sum\limits_{i\in W_{4}}+\sum\limits_{i\in W_{5}}\right)x^{i}\in \mathrm{GF}(q)[x],\label{S(x)}
\end{align}
\begin{align}
T(x)=\left(\sum\limits_{i\in P}+\sum\limits_{i\in W_{1}}+\sum\limits_{i\in W_{2}}+\sum\limits_{i\in W_{3}}\right)x^{i}\in \mathrm{GF}(q)[x] \ \mathrm{and} \label{T(x)}
\end{align}
\begin{align}
U(x)=\left(\sum\limits_{i\in P}+\sum\limits_{i\in W_{2}}+\sum\limits_{i\in W_{3}}+\sum\limits_{i\in W_{4}}\right)x^{i}\in \mathrm{GF}(q)[x].\label{U(x)}
\end{align}
Our main aim in this section is to find the generator polynomial
\begin{align}
 g(x)=\frac{x^n-1}{\mathrm{gcd}(x^n-1,S(x))}.\nonumber
\end{align}
of the cyclic code $C_s$ defined by the sequence $s^{\infty}.$ 
To compute the parameters of the cyclic code $C_s$ defined by the sequence $s^{\infty}$, we need to compute $\mathrm{gcd}(x^{n}-1,S(x))$. Since $\beta$ is a primitive $n$th root of unity, we need only to find such $t$'s that $S(\beta^{t})=0$, where $0\leq t\leq n-1$. To this end, we need number of auxiliary results.
We have 
\begin{align}
0=\beta^{n}-1=(\beta^{n_1})^{n_2}-1=(\beta^{n_1}-1)(1 + \beta^{n_1} + \beta^{2n_1} + \cdots + \beta^{(n_2-1)n_1}).\nonumber
\end{align}
It follows that
\begin{align}
 \beta^{n_1}+\beta^{2n_1}+\cdots+\beta^{(n_2-1)n_1}=-1, i.e., \sum\limits_{i\in P} \beta^{i} = -1.\label{P}
\end{align}
By symmetry we get
\begin{align}
\beta^{n_2} + \beta^{2n_2} + \cdots + \beta^{(n_1-1)n_2} = -1, i.e., \sum\limits_{i\in Q}\beta^{i}=-1.\label{Q}
\end{align}

\begin{lemma}\label{P&Q}
Let the symbols be same as before. For $0\leq j\leq5$, we have
$$
\sum\limits_{i\in W_{j}}\beta^{it}=\left\{
\begin{array}{ll}
-\frac{n_{1}-1}{6}\ (\mathrm{mod}\ p),\ \mathrm{if}\ t\in P,\\
-\frac{n_{2}-1}{6}\ (\mathrm{mod}\ p),\ \mathrm{if}\ t\in Q.
\end{array}
\right.
$$
\end{lemma}

\begin{proof} Suppose that $t\in Q$. Since $g$ is a common primitive roots of $n_1$ and $n_2$ and the order of $g$ modulo $n$ is $e$,
by the definition of $u$, we have
\begin{align}
W_{j}\ \mathrm{mod}\ n_{1}&=\{g^{s}u^{j}\ \mathrm{mod}\ n_{1}:\ s=0,1,2,\cdots,e-1\}\nonumber\\
&=\{g^{s+j}\ \mathrm{mod}\ n_{1}:\ s=0,1,2,\cdots,e-1\}\nonumber\\
&=\frac{n_{2}-1}{6}\ast\{1,2,\cdots,n_{1}-1\},\nonumber
\end{align}
where $\frac{n_{2}-1}{6}$ denotes the multiplicity of each element in the set $\{1,2,\cdots,n_{1}-1\}$.
 We can write $g^sx^j$ in the form 
\begin{align}
&1+k_{11}n_1,1+k_{12}n_1,\cdots,1+k_{1(n_2-1)/6}n_1,\nonumber \\
&2+k_{21}n_1,2+k_{22}n_1,\cdots,2+k_{2(n_2-1)/6}n_1,\nonumber \\
\vdots \nonumber \\
&n_1-1+k_{(n_1-1)1}n_1,n_1-1+k_{(n_1-1)2}n_1,\cdots,n_1-1+k_{(n_1-1)(n_2-1)/6}n_1.\label{g^sx^j}
 \end{align}
where $k_{li}$ is an positive integer, $1\leq l\leq n_1-1$ and $1\leq i\leq (n_2-1)/6$.\\
When $s$ ranges over $\{0, 1, \cdots\ , e-1\},$ we divides the set $W_j$ into $(n_2-1)/6$ subsets each of which contains ${n_1 -1}$
consecutive integers, i.e., $ \ g^{s+j}\ \mathrm{mod}\ n_{1}$ takes on each element of $\{1,2,\cdots,n_{1}-1\}$ exactly $\frac{n_{2}-1}{6}$
times. From (\ref{g^sx^j}), it follows that 
if $t\in Q,$ we have $\beta^{(m + k_{li}n_1)t}=\beta^{mt}$, where $1\leq m \leq n_1-1.$  
 It follows from (\ref{Q}) that
$$\sum\limits_{i\in W_{j}}\beta^{it}=(\frac{n_{2}-1}{6})\sum\limits_{j\in Q}\beta^{j}=-\frac{n_{2}-1}{6}\ (\mathrm{mod}\ p).$$
For $t\in P$, we can get the result by similar argument.
\end{proof}
\begin{lemma}\label{mod 6}
 For any $r\in W_{i}$, we have $rW_{j}=W_{i+j\ (\mathrm{mod}\ d)}$, where $rW_{j}=\{rt\ |\ t \in W_{j}\}$.
\end{lemma}

\begin{proof} We have $ W_{i} =\{g^{s}u^{i}:\ s=0,1,2,\cdots,e-1\},i=0,1,\cdots,d-1 $ and let $r = g^{s_1}u^{i}\in W_{i}.$ 
Then $rW_{j} =g^{s_1}u^{i}\{u^{j}+gu^{j}+g^{1}u^{j}+\cdots+g^{e-1}u^{j}\}=\{g^{s_1}u^{i+j}+g^{s_1+1}u^{i+j}+\cdots+g^{s_1+e-1}u^{i+j}\}.$
Since $u\in \mathbb{Z}_{n}^{\ast}$, there must exist an integer $\upsilon$ with $0\leqslant \upsilon \leqslant e-1$ such that $u^d=g^\upsilon$, therefore, we must have $rW_{j}=W_{i+j\ (\mathrm{mod}\ d)}.$
\end{proof}
\begin{lemma}\label{stu}
For all $t\in \mathbb{Z}_{n}$, we have
$$
S(\beta^{t})=\left\{
\begin{array}{llllll}
-\frac{n_{1}+1}{2}\ (\mathrm{mod}\ p),\ \mathrm{if}\ t\in P,\\
\ \ \frac{n_{2}-1}{2}\ (\mathrm{mod}\ p),\ \mathrm{if}\ t\in Q,\\
\ \ S(\beta),\ \ \ \ \ \ \ \ \ \ \ \ \mathrm{if}\ t\in W_{0},\\
-(T(\beta)+1),\ \ \ \ \ \mathrm{if}\ t\in W_{1},\\
-(U(\beta)+1),\ \ \ \ \ \mathrm{if}\ t\in W_{2},\\
-(S(\beta)+1),\ \ \ \ \ \mathrm{if}\ t\in W_{3},\\
\ \ T(\beta),\ \ \ \ \ \ \ \ \ \ \ \ \mathrm{if}\ t\in W_{4},\\
\ \ U(\beta),\ \ \ \ \ \ \ \ \ \ \ \ \mathrm{if}\ t\in W_{5},
\end{array}
\right.
$$
$$
T(\beta^{t})=\left\{
\begin{array}{llllll}
-\frac{n_{1}+1}{2}\ (\mathrm{mod}\ p),\ \mathrm{if}\ t\in P,\\
\ \ \frac{n_{2}-1}{2}\ (\mathrm{mod}\ p),\ \mathrm{if}\ t\in Q,\\
\ \ T(\beta),\ \ \ \ \ \ \ \ \ \ \ \ \mathrm{if}\ t\in W_{0},\\
\ \ U(\beta),\ \ \ \ \ \ \ \ \ \ \ \ \mathrm{if}\ t\in W_{1},\\
\ \ S(\beta),\ \ \ \ \ \ \ \ \ \ \ \ \mathrm{if}\ t\in W_{2},\\
-(T(\beta)+1),\ \ \ \ \ \mathrm{if}\ t\in W_{3},\\
-(U(\beta)+1),\ \ \ \ \ \mathrm{if}\ t\in W_{4},\\
-(S(\beta)+1),\ \ \ \ \ \mathrm{if}\ t\in W_{5},
\end{array}
\right.
$$
and
$$
U(\beta^{t})=\left\{
\begin{array}{llllll}
-\frac{n_{1}+1}{2}\ (\mathrm{mod}\ p),\ \mathrm{if}\ t\in P,\\
\ \ \frac{n_{2}-1}{2}\ (\mathrm{mod}\ p),\ \mathrm{if}\ t\in Q,\\
\ \ U(\beta),\ \ \ \ \ \ \ \ \ \ \ \ \mathrm{if}\ t\in W_{0},\\
\ \ S(\beta),\ \ \ \ \ \ \ \ \ \ \ \ \mathrm{if}\ t\in W_{1},\\
-(T(\beta)+1),\ \ \ \ \ \mathrm{if}\ t\in W_{2},\\
-(U(\beta)+1),\ \ \ \ \ \mathrm{if}\ t\in W_{3},\\
-(S(\beta)+1),\ \ \ \ \ \mathrm{if}\ t\in W_{4},\\
\ \ T(\beta),\ \ \ \ \ \ \ \ \ \ \ \ \mathrm{if}\ t\in W_{5}.\\
\end{array}
\right.
$$
\end{lemma}

\begin{proof} Since $\mathrm{gcd}(n_{1},n_{2})=1$, if $t\in P$ then $tP=P$ . Then by (\ref{P}) and  Lemma \ref{P&Q}, we get 
\begin{align}
S(\beta^{t})=\sum\limits_{i\in C_{1}}\beta^{ti}&=\left(\sum\limits_{i\in P}+\sum\limits_{i\in W_{3}}+\sum\limits_{i\in W_{4}}+\sum\limits_{i\in W_{5}}\right)\beta^{ti}\nonumber\\
&=(-1\ \mathrm{mod}\ p)  -\left(\frac{n_{1}-1}{6}\ \mathrm{mod}\ p\right)-\left(\frac{n_{1}-1}{6}\ \mathrm{mod}\ p\right)-\left(\frac{n_{1}-1}{6}\ \mathrm{mod}\ p\right)\nonumber\\
&=-\frac{n_{1}+1}{2}\ \mathrm{mod}\ p.\nonumber
\end{align}
If $t\in Q$, then $tP=0$. Then by  (\ref{P}) and  Lemma \ref{P&Q}, we get
\begin{align}
S(\beta^{t})=\sum\limits_{i\in C_{1}}\beta^{ti}&=\left(\sum\limits_{i\in P}+\sum\limits_{i\in W_{3}}+\sum\limits_{i\in W_{4}}+\sum\limits_{i\in W_{5}}\right)\beta^{ti}\nonumber\\
&=(n_{2}-1 \ \mathrm{mod}\ p)-\left(\frac{n_{2}-1}{6}\ \mathrm{mod}\ p\right)-\left(\frac{n_{2}-1}{6}\ \mathrm{mod}\ p\right)-\left(\frac{n_{2}-1}{6}\ \mathrm{mod}\ p\right)\nonumber\\
&=\frac{n_{2}-1}{2}\ \mathrm{mod}\ p.\nonumber
\end{align}
By Lemma \ref{mod 6}, $tW_i=W_i$ if $t\in W_0$. If $t\in W_0,$ then $tP=P$ since $\mathrm{gcd}(t,n_{2})=1$. Hence
\begin{align}
S(\beta^{t})=\sum\limits_{i\in C_{1}}\beta^{ti}&=\left(\sum\limits_{i\in P}+\sum\limits_{i\in W_{3}}+\sum\limits_{i\in W_{4}}+\sum\limits_{i\in W_{5}}\right)\beta^{ti}\nonumber\\
&=\left(\sum\limits_{i\in P}+\sum\limits_{i\in W_{3}}+\sum\limits_{i\in W_{4}}+\sum\limits_{i\in W_{5}}\right)\beta^{i}\nonumber\\
&=S(\beta).\nonumber
\end{align}
By  Lemma \ref{mod 6}, if $t\in W_1$ then $tW_i=W_{(i+1)( \mathrm{mod\ 6})}$ for $0\leq i\leq5$. And since $\mathrm{gcd}(a,n_{2})=1$, if $t\in W_1$ then $tP=P$. 
We have $\beta^n-1=(\beta-1)(\sum\limits_{i=0}^{n-1}\beta^i)=0$ and $\beta-1\neq0,$ this gives $\sum\limits_{i=0}^{n-1}\beta^i=0.$ 
Therefore, $\sum\limits_{i=0}^{n-1}\beta^i=1+\sum\limits_{i\in P}\beta^i+\sum\limits_{i\in Q}\beta^i+\sum\limits_{i\in\bigcup\limits_{j=0}^{5} W_{j}}\beta^{i}=0.$ 
From (\ref{P}) and (\ref{Q}), we get
\begin{equation}
\sum\limits_{i\in\bigcup\limits_{j=0}^{5} W_{j}}\beta^{i}=1.\label{sum}
\end{equation}
Hence
\begin{align}
S(\beta^{t})=\sum\limits_{i\in C_{1}}\beta^{ti}&=\left(\sum\limits_{i\in P}+\sum\limits_{i\in W_{3}}+\sum\limits_{i\in W_{4}}+\sum\limits_{i\in W_{5}}\right)\beta^{ti} \nonumber \\ 
&=\left(\sum\limits_{i\in P}+\sum\limits_{i\in W_{4}}+\sum\limits_{i\in W_{5}}+\sum\limits_{i\in W_{0}}\right)\beta^{i}\nonumber\\
&=\left(\sum\limits_{i\in P}-\sum\limits_{i\in W_{1}}-\sum\limits_{i\in W_{2}}-\sum\limits_{i\in W_{3}}\right)\beta^{i}+1\nonumber\\
&=\left(-\sum\limits_{i\in P}-\sum\limits_{i\in W_{1}}-\sum\limits_{i\in W_{2}}-\sum\limits_{i\in W_{3}}\right)\beta^{i}+2\sum\limits_{i\in P}\beta^{i}+1\nonumber\\
&=-(T(\beta)+1).\nonumber
\end{align}
Similarly, we can get the result when $t\in {W_i},\ \  2\leq i\leq5$.\\
In a similar fashion, we can get the result for $T(\beta^{t})$ and $U(\beta^{t}).$ This completes the proof of this lemma.
\end{proof}
Note that
\begin{align}
 S(1)=\frac{(n_1+1)(n_2-1)}{2} (\mathrm{mod \ p}).\label{val1}
\end{align}

\begin{corollary}\label{qw0}
Let the symbols be defined as before. We have the following conclusions.\\
\rm{(I)} If $q\notin W_{0}$, we have $S(\beta)\neq0,-1$, $T(\beta)\neq0,-1$ and $U(\beta)\neq0,-1.$ \\
\rm{(II)} If $q\in W_{0}$, we have $S^{q}(\beta)=S(\beta)$, $T^{q}(\beta)=T(\beta)$ and $U^{q}(\beta)=U(\beta)$ and ${S(\beta), T(\beta), U(\beta)}\in GF(q).$
\end{corollary}
\begin{proof} (I) Note that $\mathrm{gcd}(n,q)=1$, i.e., $q\in \mathbb{Z}_{n}^{\ast}$ then $q\in\bigcup\limits_{i=1}^{5}W_i$. If $q\notin W_{0}$, without loss of generality, assume that
$q\in W_{1}$. By Lemma \ref{stu}, we have
\begin{align}
S^{q^{3}}(\beta)=S^{q^{2}}(\beta^{q})=(-T(\beta)-1)^{q^{2}}=(-T(\beta^{q})-1)^{q} = (-U(\beta)-1)^{q}=(-U(\beta^{q})-1)=-S(\beta)-1,\nonumber
\end{align}
i.e.,
\begin{align}
 S^{q^{3}}(\beta)+S(\beta)+1=0.\label{S}
\end{align}
It is easy to see that $0$ and $-1$ is not the solution of Eq. (\ref{S}). Similarly, we have
$$T^{q^{3}}(\beta)+T(\beta)+1=0,$$ and
$$U^{q^{3}}(\beta)+U(\beta)+1=0,$$
i.e., $T(\beta)\neq0,-1$ and $U(\beta)\neq0,-1$.
If $q\in W_i, 2\leqslant i\leqslant 5 $ the results can be proved by similar argument.\\
(II) If $q\in W_{0}$, the conclusion is obvious.
\end{proof}
\begin{lemma}\label{STU}
Let the symbols be the same as before. We have the following conclusions.\\
\rm{(I)} If $\eta$ is odd, then we have three cases:\\
Case \rm{(A)} If $\rho- \varrho \equiv 0\ (\mathrm{mod}\ 3),$
\begin{align}
& S(\beta)(S(\beta)+1)=\frac{n-1}{4}-\frac{2y}{3}\left(-\sum\limits_{i\in W_{0}}+\sum\limits_{i\in W_{2}}-\sum\limits_{i\in W_{3}}+\sum\limits_{i\in W_{5}}\right)\beta^i, \nonumber \\
& T(\beta)(T(\beta)+1)=\frac{n-1}{4}-\frac{2y}{3}\left(\sum\limits_{i\in W_{0}}-\sum\limits_{i\in W_{1}}+\sum\limits_{i\in W_{3}}-\sum\limits_{i\in W_{4}}\right)\beta^i \ \mathrm{and} \nonumber\\
& U(\beta)(U(\beta)+1)=\frac{n-1}{4}-\frac{2y}{3}\left(\sum\limits_{i\in W_{1}}-\sum\limits_{i\in W_{2}}+\sum\limits_{i\in W_{4}}-\sum\limits_{i\in W_{5}}\right)\beta^i. \nonumber
\end{align}
Case \rm{(B)} If $\rho- \varrho \equiv 1\ (\mathrm{mod}\ 3),$
\begin{align}
& S(\beta)(S(\beta)+1)=\frac{n-1}{4}+\frac{x+y}{3}\left(\sum\limits_{i\in W_{1}}-\sum\limits_{i\in W_{2}}+\sum\limits_{i\in W_{4}}-\sum\limits_{i\in W_{5}}\right)\beta^i, \nonumber \\
& T(\beta)(T(\beta)+1)=\frac{n-1}{4}+\frac{x+y}{3}\left(-\sum\limits_{i\in W_{0}}+\sum\limits_{i\in W_{2}}-\sum\limits_{i\in W_{3}}+\sum\limits_{i\in W_{5}}\right)\beta^i\ \mathrm{and} \nonumber\\
& U(\beta)(U(\beta)+1)=\frac{n-1}{4}+\frac{x+y}{3}\left(\sum\limits_{i\in W_{0}}-\sum\limits_{i\in W_{1}}+\sum\limits_{i\in W_{3}}-\sum\limits_{i\in W_{4}}\right)\beta^i. \nonumber
\end{align}
Case \rm{(C)} If $\rho- \varrho \equiv 2\ (\mathrm{mod}\ 3), $
\begin{align}
& S(\beta)(S(\beta)+1)=\frac{n-1}{4}-\frac{x-y}{3}\left(\sum\limits_{i\in W_{0}}-\sum\limits_{i\in W_{1}}+\sum\limits_{i\in W_{3}}-\sum\limits_{i\in W_{4}}\right)\beta^i, \nonumber \\
& T(\beta)(T(\beta)+1)=\frac{n-1}{4}-\frac{x-y}{3}\left(\sum\limits_{i\in W_{1}}-\sum\limits_{i\in W_{2}}+\sum\limits_{i\in W_{4}}-\sum\limits_{i\in W_{5}}\right)\beta^i \ \mathrm{and} \nonumber\\
& U(\beta)(U(\beta)+1)=\frac{n-1}{4}-\frac{x-y}{3}\left(-\sum\limits_{i\in W_{0}}+\sum\limits_{i\in W_{2}}-\sum\limits_{i\in W_{3}}+\sum\limits_{i\in W_{5}}\right)\beta^i. \nonumber
\end{align}
\rm{(II)} If $\eta$ is even, then we have three cases:\\
Case \rm{(A)} If $\rho- \varrho \equiv 0\ (\mathrm{mod}\ 3),$
\begin{align}
& S(\beta)(S(\beta)+1)=-\frac{n+1}{4}+\frac{2y}{3}\left(-\sum\limits_{i\in W_{0}}+\sum\limits_{i\in W_{2}}-\sum\limits_{i\in W_{3}}+\sum\limits_{i\in W_{5}}\right)\beta^i, \nonumber \\
& T(\beta)(T(\beta)+1)=-\frac{n+1}{4}+\frac{2y}{3}\left(\sum\limits_{i\in W_{0}}-\sum\limits_{i\in W_{1}}+\sum\limits_{i\in W_{3}}-\sum\limits_{i\in W_{4}}\right)\beta^i \ \mathrm{and} \nonumber \\
& U(\beta)(U(\beta)+1)=-\frac{n+1}{4}+\frac{2y}{3}\left(\sum\limits_{i\in W_{1}}-\sum\limits_{i\in W_{2}}+\sum\limits_{i\in W_{4}}-\sum\limits_{i\in W_{5}}\right)\beta^i. \nonumber
\end{align}
Case \rm{(B)} If $\rho- \varrho \equiv 1\ (\mathrm{mod}\ 3),$
\begin{align}
& S(\beta)(S(\beta)+1)=-\frac{n+1}{4}-\frac{x+y}{3}\left(\sum\limits_{i\in W_{1}}-\sum\limits_{i\in W_{2}}+\sum\limits_{i\in W_{4}}-\sum\limits_{i\in W_{5}}\right)\beta^i, \nonumber \\
& T(\beta)(T(\beta)+1)=-\frac{n+1}{4}-\frac{x+y}{3}\left(-\sum\limits_{i\in W_{0}}+\sum\limits_{i\in W_{2}}-\sum\limits_{i\in W_{3}}+\sum\limits_{i\in W_{5}}\right)\beta^i \ \mathrm{and} \nonumber\\
& U(\beta)(U(\beta)+1)=-\frac{n+1}{4}-\frac{x+y}{3}\left(\sum\limits_{i\in W_{0}}-\sum\limits_{i\in W_{1}}+\sum\limits_{i\in W_{3}}-\sum\limits_{i\in W_{4}}\right)\beta^i. \nonumber
\end{align}
Case \rm{(C)} If $\rho- \varrho \equiv 2\ (\mathrm{mod}\ 3), $
\begin{align}
& S(\beta)(S(\beta)+1)=-\frac{n+1}{4}+\frac{x-y}{3}\left(\sum\limits_{i\in W_{0}}-\sum\limits_{i\in W_{1}}+\sum\limits_{i\in W_{3}}-\sum\limits_{i\in W_{4}}\right)\beta^i, \nonumber \\
& T(\beta)(T(\beta)+1)=-\frac{n+1}{4}+\frac{x-y}{3}\left(\sum\limits_{i\in W_{1}}-\sum\limits_{i\in W_{2}}+\sum\limits_{i\in W_{4}}-\sum\limits_{i\in W_{5}}\right)\beta^i \ \mathrm{and} \nonumber\\
& U(\beta)(U(\beta)+1)=-\frac{n+1}{4}+\frac{x-y}{3}\left(-\sum\limits_{i\in W_{0}}+\sum\limits_{i\in W_{2}}-\sum\limits_{i\in W_{3}}+\sum\limits_{i\in W_{5}}\right)\beta^i. \nonumber
\end{align}

\end{lemma}

\begin{proof} (I) By the definition of $S(x)$ and from (\ref{P}), we have 
\begin{align}
 S(\beta)=-1+\left(\sum\limits_{i\in W_{3}}+\sum\limits_{i\in W_{4}}+\sum\limits_{i\in W_{5}}\right)\beta^i.\nonumber
\end{align}
Then, we get
\begin{align}
S(\beta)(S(\beta+1))=-\left(\sum\limits_{i\in W_{3}}+\sum\limits_{i\in W_{4}}+\sum\limits_{i\in W_{5}}\right)\beta^i+\left(\sum\limits_{i\in W_{3}}\sum\limits_{j\in W_{3}}+\sum\limits_{i\in W_{4}}\sum\limits_{j\in W_{4}}+\sum\limits_{i\in W_{5}}\sum\limits_{j\in W_{5}}\right)\beta^{i+j}\nonumber\\
+\left( 2\sum\limits_{i\in W_{3}}\sum\limits_{j\in W_{4}}+ 2\sum\limits_{i\in W_{3}}\sum\limits_{j\in W_{5}}+ 2\sum\limits_{i\in W_{4}}\sum\limits_{j\in W_{5}}\right)\beta^{i+j}.\label{STU.1} \ \ \ \ \ \ \ \ \ \ \ \ \ \ \ \ \ \ \ \ \ \ \ \ \ \ \ \ \ \ \ \ \ \ \  
\end{align}
From Lemma \ref{-1}, if $\eta$ is odd then $-1\in W_0$ and that $-W_j=\{ -t : t\in W_j\}=W_j.$ 
\begin{align}
&\ \ \ \ \  \sum\limits_{i\in W_{3}}\sum\limits_{j\in W_{3}}\beta^{i+j}=\sum\limits_{i\in W_{3}}\sum\limits_{j\in W_{3}}\beta^{i-j}\nonumber  \\
&\ \ \ \ \ \ =|W_3|+\sum\limits_{r\in P\cup Q}d(3,3;r)\beta^{r}+(3,3)_6\sum\limits_{i\in W_{0}}\beta^{i}+(2,2)_6\sum\limits_{i\in W_{1}}\beta^{i}+(1,1)_6\sum\limits_{i\in W_{2}}\beta^{i}+(0,0)_6\sum\limits_{i\in W_{3}}\beta^{i}\nonumber \\
&\ \ \ \ \ \ \ \ \ +(5,5)_6\sum\limits_{i\in W_{4}}\beta^{i}+(4,4)_6\sum\limits_{i\in W_{5}}\beta^{i},\label{STU1}
\end{align}
\begin{align}
&\ \ \ \ \  \sum\limits_{i\in W_{4}}\sum\limits_{j\in W_{4}}\beta^{i+j}=\sum\limits_{i\in W_{4}}\sum\limits_{j\in W_{4}}\beta^{i-j}\nonumber \\
&\ \ \ \ \ \ =|W_4|+\sum\limits_{r\in P\cup Q}d(4,4;r)\beta^{r}+(4,4)_6\sum\limits_{i\in W_{0}}\beta^{i}+(3,3)_6\sum\limits_{i\in W_{1}}\beta^{i}+(2,2)_6\sum\limits_{i\in W_{2}}\beta^{i}+(1,1)_6\sum\limits_{i\in W_{3}}\beta^{i}\nonumber \\
&\ \ \ \ \ \ \ \ \ +(0,0)_6\sum\limits_{i\in W_{4}}\beta^{i}+(5,5)_6\sum\limits_{i\in W_{5}}\beta^{i},\label{STU2}
\end{align}
\begin{align}
&\ \ \ \ \  \sum\limits_{i\in W_{5}}\sum\limits_{j\in W_{5}}\beta^{i+j}=\sum\limits_{i\in W_{5}}\sum\limits_{j\in W_{5}}\beta^{i-j}\nonumber \\
&\ \ \ \ \ \ =|W_5|+\sum\limits_{r\in P\cup Q}d(5,5;r)\beta^{r}+(5,5)_6\sum\limits_{i\in W_{0}}\beta^{i}+(4,4)_6\sum\limits_{i\in W_{1}}\beta^{i}+(3,3)_6\sum\limits_{i\in W_{2}}\beta^{i}+(2,2)_6\sum\limits_{i\in W_{3}}\beta^{i}\nonumber \\
&\ \ \ \ \ \ \ \ \ +(1,1)_6\sum\limits_{i\in W_{4}}\beta^{i}+(0,0)_6\sum\limits_{i\in W_{5}}\beta^{i},\label{STU3}
\end{align}
\begin{align}
&2 \sum\limits_{i\in W_{3}}\sum\limits_{j\in W_{4}}\beta^{i+j}=\sum\limits_{i\in W_{3}}\sum\limits_{j\in W_{4}}\beta^{i-j}\nonumber \\
&= 2\ \left( \sum\limits_{r\in P\cup Q}d(4,3;r)\beta^{r}+(4,3)_6\sum\limits_{i\in W_{0}}\beta^{i}+(3,2)_6\sum\limits_{i\in W_{1}}\beta^{i}+(2,1)_6\sum\limits_{i\in W_{2}}\beta^{i}+(1,0)_6\sum\limits_{i\in W_{3}}\beta^{i}\right.\nonumber \\
&\ \ \ \left.+(0,5)_6\sum\limits_{i\in W_{4}}\beta^{i}+(5,4)_6\sum\limits_{i\in W_{5}}\beta^{i} \right) ,\label{STU4}
\end{align}
\begin{align}
&2 \sum\limits_{i\in W_{3}}\sum\limits_{j\in W_{5}}\beta^{i+j}=\sum\limits_{i\in W_{3}}\sum\limits_{j\in W_{5}}\beta^{i-j}\nonumber \\
&=2 \left( \sum\limits_{r\in P\cup Q}d(5,3;r)\beta^{r}+(5,3)_6\sum\limits_{i\in W_{0}}\beta^{i}+(4,2)_6\sum\limits_{i\in W_{1}}\beta^{i}+(3,1)_6\sum\limits_{i\in W_{2}}\beta^{i}+(2,0)_6\sum\limits_{i\in W_{3}}\beta^{i}\right. \nonumber \\
&\ \ \ \left.+(1,5)_6\sum\limits_{i\in W_{4}}\beta^{i}+(0,4)_6\sum\limits_{i\in W_{5}}\beta^{i} \right),\label{STU5}
\end{align}
\begin{align}
&2 \sum\limits_{i\in W_{4}}\sum\limits_{j\in W_{5}}\beta^{i+j}=\sum\limits_{i\in W_{4}}\sum\limits_{j\in W_{5}}\beta^{i-j}\nonumber \\
&= 2 \left(\  \sum\limits_{r\in P\cup Q}d(5,4;r)\beta^{r}+(5,4)_6\sum\limits_{i\in W_{0}}\beta^{i}+(4,3)_6\sum\limits_{i\in W_{1}}\beta^{i}+(3,2)_6\sum\limits_{i\in W_{2}}\beta^{i}+(2,1)_6\sum\limits_{i\in W_{3}}\beta^{i}\right. \nonumber \\
&\ \ \ \left.+(1,0)_6\sum\limits_{i\in W_{4}}\beta^{i}+(0,5)_6\sum\limits_{i\in W_{5}}\beta^{i} \right), \label{STU6}
\end{align}


Substituting (\ref{STU1} - \ref{STU6}) into (\ref{STU.1}) and combining Lemma \ref{value d(i,j,t)}, Lemma \ref{odd&even} and (\ref{sum}), we get
\begin{align}
S(\beta)(S(\beta)+1)=-(\sum\limits_{i\in W_{3}}\beta^{i}+\sum\limits_{i\in W_{4}}\beta^{i}+\sum\limits_{i\in W_{5}}\beta^{i}) + \left(\frac{3M}{2}-\frac{8x-4a}{72}+\frac{12b+24y}{72}\right)\sum\limits_{i\in W_0}\beta^i \nonumber\\ 
+ \left(\frac{3M}{2}-\frac{16x+8a}{72}\right)\sum\limits_{i\in W_1}\beta^i + \left(\frac{3M}{2}-\frac{8x-4a}{72}-\frac{12b+24y}{72}\right)\sum\limits_{i\in W_2}\beta^i\ \ \ \ \ \ \  \ \ \ \ \ \ \ \ \ \ \ \ \ \ \ \ \ \ \nonumber\\
+ \left(\frac{3M}{2}+1-\frac{8x-4a}{72}+\frac{12b+24y}{72}\right)\sum\limits_{i\in W_3}\beta^i+ \left(\frac{3M}{2}+1+\frac{16x+8a}{72}\right)\sum\limits_{i\in W_4}\beta^i\ \ \ \ \ \ \ \ \ \ \ \ \ \ \ \nonumber\\ + \left(\frac{3M}{2}+1-\frac{8x-4a}{72}-\frac{24y-12b}{72}\right)\sum\limits_{i\in W_5}\beta^i-12\frac{(n_1-1)(n_2-1)}{36}-3\frac{(n_1-1)(n_2-7)}{36}\ \ \ \nonumber\\
-3\frac{(n_1-7)(n_2-1)}{36}+3\frac{(n_1-1)(n_2-1)}{6}\ \ \ \ \ \ \ \ \ \ \ \ \ \ \ \ \ \ \ \ \ \ \ \ \ \ \ \  \ \ \ \ \ \ \ \ \ \ \  \ \ \ \ \ \ \ \ \ \ \ \ \ \ \ \ \ \ \ \ \ \ \  \ \ \nonumber\\
=\frac{n-1}{4}-\frac{a+2x}{18}\left(\sum\limits_{i\in W_0}-2\sum\limits_{i\in W_{1}}+\sum\limits_{i\in W_{2}}+\sum\limits_{i\in W_{3}}-2\sum\limits_{i\in W_{4}}+\sum\limits_{i\in W_{5}}\right)\beta^{i}\ \ \ \ \ \ \ \ \ \ \ \ \ \ \ \ \ \ \ \ \ \  \ \ \ \ \ \ \ \ \ \ \ \nonumber\\ 
+\frac{b+2y}{6}\left(\sum\limits_{i\in W_0}-\sum\limits_{i\in W_{2}}+\sum\limits_{i\in W_{3}}-\sum\limits_{i\in W_{5}}\right)\beta^{i}.\ \ \ \ \ \ \ \ \ \ \ \ \ \ \ \ \  \ \ \  \ \ \ \ \ \ \ \ \ \ \ \ \ \ \ \ \ \ \ \ \ \ \ \ \ \ \ \ \ \ \ \ \ \ \ \ \ \ \ \ \ \ \ \ \  \ \nonumber
\end{align}
By the same argument as above, we can get
\begin{align}
& T(\beta)(T(\beta)+1)=\frac{n-1}{4}-\frac{a+2x}{18}\left(\sum\limits_{i\in W_0}+\sum\limits_{i\in W_{1}}-2\sum\limits_{i\in W_{2}}+\sum\limits_{i\in W_{3}}+\sum\limits_{i\in W_{4}}-2\sum\limits_{i\in W_{5}}\right)\beta^{i}\nonumber\\
&~~~~~~~~~~~~~~~~~~~~~~~~~~~~~~~~+\frac{b+2y}{6}\left(-\sum\limits_{i\in W_0}+\sum\limits_{i\in W_{1}}-\sum\limits_{i\in W_{3}}+\sum\limits_{i\in W_{4}}\right)\beta^{i} \  \ \mathrm{and}\nonumber
\end{align}
\begin{align}
& U(\beta)(U(\beta)+1)=\frac{n-1}{4}-\frac{a+2x}{18}\left(-2\sum\limits_{i\in W_0}+\sum\limits_{i\in W_{1}}+\sum\limits_{i\in W_{2}}-2\sum\limits_{i\in W_{3}}+\sum\limits_{i\in W_{4}}+\sum\limits_{i\in W_{5}}\right)\beta^{i}\nonumber\\
&~~~~~~~~~~~~~~~~~~~~~~~~~~~~~~~~+\frac{b+2y}{6}\left(-\sum\limits_{i\in W_1}+\sum\limits_{i\in W_{2}}-\sum\limits_{i\in W_{4}}+\sum\limits_{i\in W_{5}}\right)\beta^{i}.\nonumber
\end{align}

(II) From the Lemma \ref{-1}, if $\eta$ is even then $-1 \in W_3$ and that $-W_j=\{-t\ : \ t\in W_j\}=W_{(j+3)\ \mathrm{mod\ 6}}$, we get
\begin{align}
& S(\beta)(S(\beta)+1)=-\frac{n-1}{4}-\frac{1}{2}-\frac{a-2x}{18}\left(\sum\limits_{i\in W_0}-2\sum\limits_{i\in W_{1}}+\sum\limits_{i\in W_{2}}+\sum\limits_{i\in W_{3}}-2\sum\limits_{i\in W_{4}}+\sum\limits_{i\in W_{5}}\right)\beta^{i}\nonumber\\
&~~~~~~~~~~~~~~~~~~~~~~~~~~~~~~~~+\frac{b-2y}{6}\left(\sum\limits_{i\in W_0}-\sum\limits_{i\in W_{2}}+\sum\limits_{i\in W_{3}}-\sum\limits_{i\in W_{5}}\right)\beta^{i},\nonumber
\end{align}

\begin{align}
& T(\beta)(T(\beta)+1)=-\frac{n-1}{4}-\frac{1}{2}-\frac{a-2x}{18}\left(\sum\limits_{i\in W_0}+\sum\limits_{i\in W_{1}}-2\sum\limits_{i\in W_{2}}+\sum\limits_{i\in W_{3}}+\sum\limits_{i\in W_{4}}-2\sum\limits_{i\in W_{5}}\right)\beta^{i}\nonumber\\
&~~~~~~~~~~~~~~~~~~~~~~~~~~~~~~~~+\frac{b-2y}{6}\left(-\sum\limits_{i\in W_0}+\sum\limits_{i\in W_{1}}-\sum\limits_{i\in W_{3}}+\sum\limits_{i\in W_{4}}\right)\beta^{i} \ \ \mathrm{and}\nonumber
\end{align}

\begin{align}
& U(\beta)(U(\beta)+1)=-\frac{n-1}{4}-\frac{1}{2}-\frac{a-2x}{18}\left(-2\sum\limits_{i\in W_0}+\sum\limits_{i\in W_{1}}+\sum\limits_{i\in W_{2}}-2\sum\limits_{i\in W_{3}}+\sum\limits_{i\in W_{4}}+\sum\limits_{i\in W_{5}}\right)\beta^{i}\nonumber\\
&~~~~~~~~~~~~~~~~~~~~~~~~~~~~~~~~+\frac{b-2y}{6}\left(-\sum\limits_{i\in W_1}+\sum\limits_{i\in W_{2}}-\sum\limits_{i\in W_{4}}+\sum\limits_{i\in W_{5}}\right)\beta^{i}.\nonumber
\end{align}
From Lemma \ref{relation}, Table 1 and Table 2, we get the desired result.\\
\end{proof}

We need to discuss the factorization of $x^n-1$ over $\mathrm{GF}(q).$
Let $\beta$ be the same as before. Define for each $i$; $0 \leq i \leq 5 $,
\begin{align}
\omega_i(x)=\prod\limits_{j\in W_i}(x-\beta^{j}),\nonumber
\end{align}
where $W_i$ denote the Whiteman's cyclotomic classes of order 6.
Among the $nth$ roots of unity $\beta^{i}$, where $0 \leq i \leq n-1$, the $n_2$ elements $\beta^{i}, i \in P \cup \{0\},$ are $n_2$th roots of unity, the $n_1$ elements $\beta^{i}, i \in Q \cup \{0\},$ are $n_1$th roots of unity. Hence,\\
$$x^{n_2}-1=\prod\limits_{i\in P \cup \{0\}}(x-\beta^{i})$$ $\ \mathrm{and}$ $$x^{n_1}-1=\prod\limits_{i\in Q \cup \{0\}}(x-\beta^{i}).$$ \\
Then, we have
$ x^n-1 = \prod\limits_{i=0}^{n-1}(x-\beta^{i}) = \frac{(x^{n_1}-1)(x^{n_2}-1)}{x-1}\omega(x),$ where $\omega(x) = \prod\limits_{i=0}^{5} \omega_i(x).$
It is straightforward to prove that if $q \in W_0 $ then $\omega_i(x) \in \mathrm{GF}(q)$  for all $i .$\\  
 
Let $\bigtriangleup_{1}=\frac{n_{1}+1}{2}\ (\mathrm{mod}\ p)$,\  $\bigtriangleup_{2}=\frac{n_{2}-1}{2}\ (\mathrm{mod}\ p)$ and 
$\bigtriangleup=\frac{(n_{1}+1)(n_{2}-1)}{2}\ (\mathrm{mod}\ p)$. From Corollary \ref{qw0}, we have the following theorems.
First, we derive the expression of generator polynomial $g(x)$ for the case $q\notin W_0.$
\begin{theorem}\label{gennotw0}

Let the symbols be defined as before and assume that $q\notin W_{0}$. 
Then the generator polynomial $g(x)$ of the sequence $s^{\infty}$ (defined in (\ref{si}))  is expressed as
$$
g(x)=\left\{
\begin{array}{llll}
x^{n}-1,\ \ \ \ \ \ \ \ \ \ \mathrm{if}\ \bigtriangleup_{1}\neq0,\ \bigtriangleup_{2}\neq0,\ \bigtriangleup\neq0\\
\frac{x^{n}-1}{x-1},\ \ \ \ \ \ \ \ \ \ \  \mathrm{if}\ \bigtriangleup_{1}\neq0,\ \bigtriangleup_{2}\neq0,\ \bigtriangleup=0\\
\frac{x^{n}-1}{x^{n_{2}-1}},\ \ \ \ \ \ \ \ \ \ \ \mathrm{if}\ \bigtriangleup_{1}=0,\ \bigtriangleup_{2}\neq0\\
\frac{x^{n}-1}{x^{n_{1}-1}},\ \ \ \ \ \ \ \ \ \ \ \mathrm{if}\ \bigtriangleup_{1}\neq0,\ \bigtriangleup_{2}=0\\
\frac{(x^{n}-1)(x-1)}{(x^{n_{1}-1})(x^{n_{2}}-1)},\ \mathrm{if}\ \bigtriangleup_{1}=\bigtriangleup_{2}=0.
\end{array}
\right.
$$
The linear span of the sequence $s^{\infty}$ is equal to $\mathrm{deg}(g(x))$. 
In this case, the cyclic code $C_s$ over $GF(q)$ defined by the two-prime WGCS-II of order 6 (defined in (\ref{si})) has generator polynomial $g(x)$ as above and parameters $[n,k,d]$, where the dimension $k=n-deg(g(x))$.
\end{theorem}

\begin{proof} If $q\notin W_{0}$. Then, from Part I of Corollary \ref{qw0}, we have $S(\beta),T(\beta) \ \mathrm{and}\ U(\beta)\neq0,-1$.
Therefore, from Lemma \ref{stu}, $S({\beta}^{t})=0$ only when t is in P or Q or both. So, from (\ref{val1}) and Lemma \ref{stu} the generator polynomial of the cyclic code $C_s$ defined by $s^{\infty}$ is expressed as above.
\end{proof}
The following theorem give the expression of the generator polynomials $g_i(x)$, $1\leqslant i \leqslant 5$ for the case $q\in W_0.$
 
 \begin{theorem}\label{genw0}
Let the symbols be defined as before. Let $q\in W_{0}$. Then we have the following results.\\
\rm{(I)} For $\bigtriangleup_{1}\neq0$, $\bigtriangleup_{2}\neq0$ and $\bigtriangleup\neq0$, 
let $g_1(x)$ denote the generator polynomial of cyclic code generated by the two-prime WGCS-II with order 6 (defined in (\ref{si})). Then we have
$$
g_{1}(x)=\left\{
\begin{array}{lllllllll}
&x^{n}-1\ \ \ \ \ \ \ \ \ \  \ \ \mathrm{if}\ \ S(\beta)\neq{0,-1},\ T(\beta)\neq{0,-1},\ U(\beta)\neq{0,-1},\\
&\frac{x^n-1}{\omega_{0}(x)}\ \ \ \ \ \ \  \ \ \ \ \ \ \ \mathrm{if}\ \ S(\beta)=0,\ T(\beta)\neq{0,-1},\ U(\beta)\neq{0,-1},\\
&\frac{x^n-1}{\omega_{3}(x)}\ \ \ \ \ \ \ \ \ \ \ \  \ \ \mathrm{if}\ \ S(\beta)=-1,\ T(\beta)\neq{0,-1},\ U(\beta)\neq{0,-1},\\
&\frac{x^n-1}{\omega_{4}(x)}\ \ \ \ \ \ \  \ \ \  \ \ \ \ \mathrm{if}\ \ T(\beta)=0,\ S(\beta)\neq{0,-1},\ U(\beta)\neq{0,-1},\\
&\frac{x^n-1}{\omega_{1}(x)}\ \ \ \ \ \ \ \  \ \ \ \ \ \  \mathrm{if}\ \ T(\beta)=-1,\ S(\beta)\neq{0,-1},\ U(\beta)\neq{0,-1},\\
&\frac{x^n-1}{\omega_{5}(x)}\ \ \ \ \ \ \ \ \  \   \  \ \ \ \mathrm{if}\ \ U(\beta)=0,\ S(\beta)\neq{0,-1},\ T(\beta)\neq{0,-1},\\
&\frac{x^n-1}{\omega_{2}(x)}\ \ \ \ \ \ \ \ \  \  \ \ \ \  \mathrm{if}\ \ U(\beta)=-1,\ S(\beta)\neq{0,-1},\ T(\beta)\neq{0,-1},\\
&\frac{x^n-1}{\omega_{0}(x)\omega_{4}(x)}\ \ \  \ \ \ \ \  \mathrm{if}\ \ S(\beta)=\ T(\beta)=0,\ U(\beta)\neq{0,-1},\\
&\frac{x^n-1}{\omega_{0}(x)\omega_{1}(x)}\ \ \ \  \ \ \ \  \mathrm{if}\ \ S(\beta)=0,\ T(\beta)=-1,\ U(\beta)\neq{0,-1},\\
&\frac{x^n-1}{\omega_{3}(x)\omega_{4}(x)}\ \ \ \  \ \ \  \ \mathrm{if}\ \ S(\beta)=-1,\ T(\beta)=0,\ U(\beta)\neq{0,-1},\\
&\frac{x^n-1}{\omega_{3}(x)\omega_{1}(x)}\ \ \ \  \ \  \ \ \mathrm{if}\ \ S(\beta)=\ T(\beta)=-1,\ U(\beta)\neq{0,-1},\\
&\frac{x^n-1}{\omega_{0}(x)\omega_{5}(x)}\ \ \ \  \ \  \ \ \mathrm{if}\ \ S(\beta)=\ U(\beta)=0,\ T(\beta)\neq{0,-1},\\
&\frac{x^n-1}{\omega_{0}(x)\omega_{2}(x)}\ \ \  \ \ \  \ \ \mathrm{if}\ \ S(\beta)=0,\ U(\beta)=-1,\ T(\beta)\neq{0,-1},\\
&\frac{x^n-1}{\omega_{3}(x)\omega_{5}(x)}\ \  \ \  \ \ \ \ \mathrm{if}\ \ S(\beta)=-1,\ U(\beta)=0,\ T(\beta)\neq{0,-1},\\
&\frac{x^n-1}{\omega_{3}(x)\omega_{2}(x)}\  \ \ \ \ \ \ \  \mathrm{if}\ \ S(\beta)=\ U(\beta)=-1,\ T(\beta)\neq{0,-1},\\
&\frac{x^n-1}{\omega_{4}(x)\omega_{5}(x)}\  \ \ \ \ \ \ \  \mathrm{if}\ \ T(\beta)=\ U(\beta)=0,\ S(\beta)\neq{0,-1},\\
&\frac{x^n-1}{\omega_{4}(x)\omega_{2}(x)}\ \  \ \ \ \ \ \  \mathrm{if}\ \ T(\beta)=0,\ U(\beta)=-1,\ S(\beta)\neq{0,-1},\\
&\frac{x^n-1}{\omega_{1}(x)\omega_{5}(x)}\ \ \ \  \ \  \ \ \mathrm{if}\ \ T(\beta)=-1,\ U(\beta)=0,\ S(\beta)\neq{0,-1},\\
&\frac{x^n-1}{\omega_{1}(x)\omega_{2}(x)}\ \ \ \ \ \ \ \   \mathrm{if}\ \ T(\beta)=\ U(\beta)=-1,\ S(\beta)\neq{0,-1},\\
&\frac{x^n-1}{\omega_{0}(x)\omega_{4}(x)\omega_{5}(x)}\ \ \mathrm{if}\ \ S(\beta)=T(\beta)=\ U(\beta)=0,\\
&\frac{x^n-1}{\omega_{0}(x)\omega_{4}(x)\omega_{2}(x)}\ \ \mathrm{if}\ \ S(\beta)=T(\beta)=0,\ U(\beta)=-1,\\
&\frac{x^n-1}{\omega_{0}(x)\omega_{1}(x)\omega_{5}(x)}\ \ \mathrm{if}\ \ S(\beta)=0,\ T(\beta)=-1\ U(\beta)=0,\\
&\frac{x^n-1}{\omega_{0}(x)\omega_{1}(x)\omega_{2}(x)}\ \ \mathrm{if}\ \ S(\beta)=0,\ T(\beta)= U(\beta)=-1,\\
&\frac{x^n-1}{\omega_{3}(x)\omega_{4}(x)\omega_{5}(x)}\ \ \mathrm{if}\ \ S(\beta)=-1,\ T(\beta)= U(\beta)=0,\\
&\frac{x^n-1}{\omega_{3}(x)\omega_{4}(x)\omega_{2}(x)}\ \ \mathrm{if}\ \ S(\beta)=-1,\ T(\beta)=0,\ U(\beta)=-1,\\
&\frac{x^n-1}{\omega_{3}(x)\omega_{1}(x)\omega_{5}(x)}\ \ \mathrm{if}\ \ S(\beta)=T(\beta)=-1,\ U(\beta)=0,\\
&\frac{x^n-1}{\omega_{3}(x)\omega_{1}(x)\omega_{2}(x)}\ \ \mathrm{if}\ \ S(\beta)=T(\beta)= U(\beta)=-1.\\
\end{array}
\right.
$$
\rm{(II)} For $\bigtriangleup_{1}\neq0$, $\bigtriangleup_{2}\neq0$ and $\bigtriangleup=0$, let $g_{2}(x)$ denote the generator 
polynomial of cyclic code generated by the two-prime WGCS-II of order 6 (defined in (\ref{si})). And let $g_{1}(x)$ be the same as in \rm{(I)}.
Then we have $g_{2}(x)=\frac{g_{1}(x)}{x-1}.$\\
\rm{(III)} For $\bigtriangleup_{1}\neq0$ and $\bigtriangleup_{2}=0$, let $g_{3}(x)$ denote the generator polynomial of cyclic code
generated by the two-prime WGCS-II of order 6 (defined in (\ref{si})). And let $g_{1}(x)$ be the same as in \rm{(I)}. Then we have $g_{3}(x)=\frac{g_{1}(x)}{x^{n_{1}}-1}.$\\
\rm{(IV)} For $\bigtriangleup_{2}\neq0$ and $\bigtriangleup_{1}=0$, let $g_{4}(x)$ denote the generator polynomial of cyclic code
generated by the  two-prime WGCS-II of order 6 (defined in (\ref{si})). And let $g_{1}(x)$ be the same as in \rm{(I)}. Then we have $g_{4}(x)=\frac{g_{1}(x)}{x^{n_{2}}-1}.$\\
\rm{(V)} For $\bigtriangleup_{1}=\bigtriangleup_{2}=0$, let $g_{5}(x)$ denote the generator polynomial of cyclic code generated by 
the two-prime WGCS-II of order 6 (defined in (\ref{si})). And let $g_{1}(x)$ be the same as in \rm{(I)}. Then we have $g_{5}(x)=\frac{g_{1}(x)(x-1)}{(x^{n_{1}}-1)(x^{n_{2}}-1)}.$
\end{theorem}

\begin{proof}  From Part II of Corollary \ref{qw0}, if $q\in W_0$ then we have $S(\beta),\ T(\beta),\ U(\beta)\in \mathrm{GF}(q).$ 
Hence, it is possible that $S(\beta)\in \{0,-1\},\ T(\beta)\in \{0,-1\},$ and $U(\beta)\in \{0,-1\}.$ The conclusion on the generator polynomial $g_1(x)$
of cyclic code generated by the two-prime WGCS-II of order 6 follows from Theorem \ref{gennotw0} and Lemma \ref{stu}.\\
We give the following corollary for $S(\beta)(S(\beta)+1)=0,$ $T(\beta)(T(\beta)+1)=0$ and $U(\beta)(U(\beta)+1)=0.$ 
\end{proof}
\begin{corollary}\label{condition}
 Let the symbols be defined as before. We have the following conclusions for $q \in W_0$.\\
 When $n\equiv 1 (\mathrm{mod} \ 12)$ and $\frac{n-1}{4}\equiv 0 (\mathrm{mod} \ p)$ or $n\equiv 7 (\mathrm{mod} \ 12)$ and $\frac{n+1}{4}\equiv 0 (\mathrm{mod} \ p)$ then the generator polynomial $g_1(x)$ (defined as above) is expressed as:\\
 Case $\mathrm{(I)}$ If $\rho - \varrho\equiv 0\ (\mathrm{mod} \ 3),$
 \[
g_{1}(x)=\left\{
\begin{array}{lllllllll}
&\frac{x^n-1}{\omega_{0}(x)\omega_{4}(x)\omega_{5}(x)},\ \ \mathrm{if}\ \frac{2y}{3}\equiv 0 \ (\mathrm{mod} \ p) \ \mathrm{and} \ \ S(\beta)=T(\beta)=\ U(\beta)=0,\\
&\frac{x^n-1}{\omega_{0}(x)\omega_{4}(x)\omega_{2}(x)},\ \ \mathrm{if}\ \frac{2y}{3}\equiv 0 \ (\mathrm{mod} \ p) \ \mathrm{and} \ \ S(\beta)=T(\beta)=0,\ U(\beta)=-1,\\
&\frac{x^n-1}{\omega_{0}(x)\omega_{1}(x)\omega_{5}(x)},\ \ \mathrm{if}\ \frac{2y}{3}\equiv 0 \ (\mathrm{mod} \ p) \ \mathrm{and} \ \ S(\beta)=0,\ T(\beta)=-1\ U(\beta)=0,\\
&\frac{x^n-1}{\omega_{0}(x)\omega_{1}(x)\omega_{2}(x)},\ \ \mathrm{if}\ \frac{2y}{3}\equiv 0 \ (\mathrm{mod} \ p) \ \mathrm{and} \ \ S(\beta)=0,\ T(\beta)= U(\beta)=-1,\\
&\frac{x^n-1}{\omega_{3}(x)\omega_{4}(x)\omega_{5}(x)},\ \ \mathrm{if}\ \frac{2y}{3}\equiv 0 \ (\mathrm{mod} \ p) \ \mathrm{and} \ \ S(\beta)=-1,\ T(\beta)= U(\beta)=0,\\
&\frac{x^n-1}{\omega_{3}(x)\omega_{4}(x)\omega_{2}(x)},\ \ \mathrm{if}\ \frac{2y}{3}\equiv 0 \ (\mathrm{mod} \ p) \ \mathrm{and} \ \ S(\beta)=-1,\ T(\beta)=0,\ U(\beta)=-1,\\
&\frac{x^n-1}{\omega_{3}(x)\omega_{1}(x)\omega_{5}(x)},\ \ \mathrm{if}\ \frac{2y}{3}\equiv 0 \ (\mathrm{mod} \ p) \ \mathrm{and} \ \ S(\beta)=T(\beta)=-1,\ U(\beta)=0,\\
&\frac{x^n-1}{\omega_{3}(x)\omega_{1}(x)\omega_{2}(x)},\ \ \mathrm{if}\ \frac{2y}{3}\equiv 0 \ (\mathrm{mod} \ p) \ \mathrm{and} \ \ S(\beta)=T(\beta)= U(\beta)=-1.\\
\end{array} \right.\\
\]
 Case $\mathrm{(II)}$ If $\rho - \varrho\equiv 1\ (\mathrm{mod} \ 3),$
 $$
g_{1}(x)=\left\{
\begin{array}{lllllllll}
&\frac{x^n-1}{\omega_{0}(x)\omega_{4}(x)\omega_{5}(x)},\ \ \mathrm{if}\ \frac{x+y}{3}\equiv 0 \ (\mathrm{mod} \ p) \ \mathrm{and} \ \ S(\beta)=T(\beta)=\ U(\beta)=0,\\
&\frac{x^n-1}{\omega_{0}(x)\omega_{4}(x)\omega_{2}(x)},\ \ \mathrm{if}\ \frac{x+y}{3}\equiv 0 \ (\mathrm{mod} \ p) \ \mathrm{and} \ \ S(\beta)=T(\beta)=0,\ U(\beta)=-1,\\
&\frac{x^n-1}{\omega_{0}(x)\omega_{1}(x)\omega_{5}(x)},\ \ \mathrm{if}\ \frac{x+y}{3}\equiv 0 \ (\mathrm{mod} \ p) \ \mathrm{and} \ \ S(\beta)=0,\ T(\beta)=-1\ U(\beta)=0,\\
&\frac{x^n-1}{\omega_{0}(x)\omega_{1}(x)\omega_{2}(x)},\ \ \mathrm{if}\ \frac{x+y}{3}\equiv 0 \ (\mathrm{mod} \ p) \ \mathrm{and} \ \ S(\beta)=0,\ T(\beta)= U(\beta)=-1,\\
&\frac{x^n-1}{\omega_{3}(x)\omega_{4}(x)\omega_{5}(x)},\ \ \mathrm{if}\ \frac{x+y}{3}\equiv 0 \ (\mathrm{mod} \ p) \ \mathrm{and} \ \ S(\beta)=-1,\ T(\beta)= U(\beta)=0,\\
&\frac{x^n-1}{\omega_{3}(x)\omega_{4}(x)\omega_{2}(x)},\ \ \mathrm{if}\ \frac{x+y}{3}\equiv 0 \ (\mathrm{mod} \ p) \ \mathrm{and} \ \ S(\beta)=-1,\ T(\beta)=0,\ U(\beta)=-1,\\
&\frac{x^n-1}{\omega_{3}(x)\omega_{1}(x)\omega_{5}(x)},\ \ \mathrm{if}\ \frac{x+y}{3}\equiv 0 \ (\mathrm{mod} \ p) \ \mathrm{and} \ \ S(\beta)=T(\beta)=-1,\ U(\beta)=0,\\
&\frac{x^n-1}{\omega_{3}(x)\omega_{1}(x)\omega_{2}(x)},\ \ \mathrm{if}\ \frac{x+y}{3}\equiv 0 \ (\mathrm{mod} \ p) \ \mathrm{and} \ \ S(\beta)=T(\beta)= U(\beta)=-1.\\
\end{array}
\right.
$$
Case $\mathrm{(III)}$ If $\rho - \varrho\equiv 2\ (\mathrm{mod} \ 3),$
 $$
g_{1}(x)=\left\{
\begin{array}{lllllllll}
&\frac{x^n-1}{\omega_{0}(x)\omega_{4}(x)\omega_{5}(x)},\ \ \mathrm{if}\ \frac{x-y}{3}\equiv 0 \ (\mathrm{mod} \ p) \ \mathrm{and} \ \ S(\beta)=T(\beta)=\ U(\beta)=0,\\
&\frac{x^n-1}{\omega_{0}(x)\omega_{4}(x)\omega_{2}(x)},\ \ \mathrm{if}\ \frac{x-y}{3}\equiv 0 \ (\mathrm{mod} \ p) \ \mathrm{and} \ \ S(\beta)=T(\beta)=0,\ U(\beta)=-1,\\
&\frac{x^n-1}{\omega_{0}(x)\omega_{1}(x)\omega_{5}(x)},\ \ \mathrm{if}\ \frac{x-y}{3}\equiv 0 \ (\mathrm{mod} \ p) \ \mathrm{and} \ \ S(\beta)=0,\ T(\beta)=-1\ U(\beta)=0,\\
&\frac{x^n-1}{\omega_{0}(x)\omega_{1}(x)\omega_{2}(x)},\ \ \mathrm{if}\ \frac{x-y}{3}\equiv 0 \ (\mathrm{mod} \ p) \ \mathrm{and} \ \ S(\beta)=0,\ T(\beta)= U(\beta)=-1,\\
&\frac{x^n-1}{\omega_{3}(x)\omega_{4}(x)\omega_{5}(x)},\ \ \mathrm{if}\ \frac{x-y}{3}\equiv 0 \ (\mathrm{mod} \ p) \ \mathrm{and} \ \ S(\beta)=-1,\ T(\beta)= U(\beta)=0,\\
&\frac{x^n-1}{\omega_{3}(x)\omega_{4}(x)\omega_{2}(x)},\ \ \mathrm{if}\ \frac{x-y}{3}\equiv 0 \ (\mathrm{mod} \ p) \ \mathrm{and} \ \ S(\beta)=-1,\ T(\beta)=0,\ U(\beta)=-1,\\
&\frac{x^n-1}{\omega_{3}(x)\omega_{1}(x)\omega_{5}(x)},\ \ \mathrm{if}\ \frac{x-y}{3}\equiv 0 \ (\mathrm{mod} \ p) \ \mathrm{and} \ \ S(\beta)=T(\beta)=-1,\ U(\beta)=0,\\
&\frac{x^n-1}{\omega_{3}(x)\omega_{1}(x)\omega_{2}(x)},\ \ \mathrm{if}\ \frac{x-y}{3}\equiv 0 \ (\mathrm{mod} \ p) \ \mathrm{and} \ \ S(\beta)=T(\beta)= U(\beta)=-1.\\
\end{array}
\right.
$$

 For $j=2,3,4,5,$ the generator polynomials $g_j(x)$ (defined as in Theorem 2) can be expressed in a similar fashion as above (as for $g_1(x)$) for $\eta$ is even and odd.
\end{corollary}

\begin{proof} From Lemma 4, if $\eta$ is odd then $n\equiv 1 \ (\mathrm{mod}\ 12).$ Let $n\equiv 1 \ (\mathrm{mod}\ 12)$ and $(n-1)/4 \equiv 0\   (\mathrm{mod}\ p)$. 
By Lemma \ref{relation}, we have
\begin{center}
if $\rho-\varrho\equiv 0\ (\mathrm{mod}\ 3),$ then $(2y)/3$ is an integer,\ \ \ \ \ \\
\ \ \ \ \ \ if $\rho-\varrho\equiv 1\ (\mathrm{mod}\ 3),$ then $(x+y)/3$ is an integer and \\
if $\rho-\varrho\equiv 2\ (\mathrm{mod}\ 3),$ then $(x-y)/3$ is an integer.
\end{center}
By the help of Lemma \ref{STU} and Theorem \ref{genw0}, we get the desired result on the generator polynomial $g_1(x)$
of cyclic code generated by two-prime WGCS-II with order 6. In this case, the cyclic code $C_s$ over $\mathrm{GF}(q)$ defined by sequence
$s^\infty$ has parameters $[n,k,d],$ where the dimension $k=n-\mathrm{deg}(g_1(x)).$
In a similar fashion, we get the result for $n\equiv 7\ (\mathrm{mod}\ 12)$  and $\frac{n+1}{4}\equiv 0 \ (\mathrm{mod} \ p).$\\
\end{proof}
\begin{remark*} We discuss the cases  for $\frac{2y}{3}\ \ \mathrm{mod}\ p\neq 0, \frac{x+y}{3}\ \ \mathrm{mod}\ p\neq0\ \ \ \mathrm{and}\ \ \ \frac{x-y}{3}\ \ \mathrm{mod}\ p\neq0$ in the above corollary.
Let $C_0= \left(\sum\limits_{i\in W_0}+\sum\limits_{i\in W_3}\right)\beta^i$, $C_1= \left(\sum\limits_{i\in W_1}+\sum\limits_{i\in W_4}\right)\beta^i$
and $C_2= \left(\sum\limits_{i\in W_2}+\sum\limits_{i\in W_5}\right)\beta^i.$ When $n\equiv 1 \ (\mathrm{mod}\ 12)$ and $\frac{n-1}{4}\equiv0\ (\mathrm{mod}\ p)$
or $n\equiv 7\ (\mathrm{mod}\ 12)$ and $\frac{n+1}{4}\equiv0 \ (\mathrm{mod}\ p).$ We have following three cases:\\
(I) If $\rho- \varrho \equiv 0\ (\mathrm{mod}\ 3)$ and $\frac{2y}{3}\ \ \mathrm{mod}\ p\neq 0$ , (II) $\rho- \varrho \equiv 1\ (\mathrm{mod}\ 3)$ and $\frac{x+y}{3}\ \ \mathrm{mod}\ p\neq 0$ 
and (III) $\rho- \varrho \equiv 2\ (\mathrm{mod}\ 3)$ and $\frac{x-y}{3}\ \ \mathrm{mod}\ p\neq 0.$\\ 
(I) If $\rho- \varrho \equiv 0\ (\mathrm{mod}\ 3)$ and $\frac{2y}{3}\ \ \mathrm{mod}\ p\neq 0,$ then we have \\
\begin{tabular}{ll}
$S(\beta)(S(\beta)+1)=0$ & if $C_2-C_0=0, C_0-C_1\neq 0, C_1-C_2\neq 0,$\\
$T(\beta)(T(\beta)+1)=0$ & if $C_0-C_1=0, C_2-C_0\neq 0, C_1-C_2\neq 0,$\\
$U(\beta)(U(\beta)+1)=0$ & if $C_1-C_2=0, C_2-C_0\neq 0, C_0-C_1\neq 0,$\\ 
$S(\beta)(S(\beta)+1)=0$  and $T(\beta)(T(\beta)+1)=0$ &  if $C_2-C_0=0,C_0-C_1= 0, C_1-C_2\neq 0,$ \\
$S(\beta)(S(\beta)+1)=0$  and $U(\beta)(U(\beta)+1)=0$ &  if $C_2-C_0=0, C_1-C_2= 0, C_0-C_1\neq 0,$\\
$T(\beta)(T(\beta)+1)=0$  and $U(\beta)(U(\beta)+1)=0$ &  if $C_0-C_1=0, C_1-C_2= 0, C_2-C_0\neq 0,$\\
$S(\beta)(S(\beta)+1)=0$, $T(\beta)(T(\beta)+1)=0$ &  if $C_2-C_0=0,C_0-C_1=0, C_1-C_2=0.$\\
and $U(\beta)(U(\beta)+1)=0$ & \\
\end{tabular}
\newline
From the above expressions, we get the similar result on generator polynomials as in Theorem \ref{genw0} with condition on 
$S(\beta),\ U(\beta),\ T(\beta),\ C_0-C_1,\ C_1-C_2$ and $C_2-C_0.$
Similarly, we get the condition for $\rho- \varrho \equiv 1\ (\mathrm{mod}\ 3)$ and $\frac{x+y}{3}\ \ \mathrm{mod}\ p\neq 0$ and 
for $\rho- \varrho \equiv 2\ (\mathrm{mod}\ 3)$ and $\frac{x-y}{3}\ \ \mathrm{mod}\ p\neq 0.$\\
\end{remark*}

\section{The minimum distance of the cyclic codes}\label{section 5}
In this section, we determine the lower bounds on the minimum distance of some of the cyclic codes of this paper.

\begin{theorem}\label{notinw0}
\cite{ding12} Let $C_{i}$ denote the cyclic code over $\mathrm{GF}(q)$ with the generator polynomial $g_{i}(x)=\frac{x^{n}-1}{x^{n_{i}}-1}.$ 
The cyclic code $C_{i}$  has parameters $[n,n_{i},d_{i}]$, where $d_{i}=n_{i-(-1)^{i}}$ and $i=1,2$.
\end{theorem}

\begin{theorem}\label{notinw0.1}
\cite{ding12} Let $C_{(n_{1},n_{2},q)}$ denote the cyclic code over $\mathrm{GF}(q)$ with the generator polynomial $g(x)=\frac{(x^{n}-1)(x-1)}{(x^{n_{1}}-1)(x^{n_{2}}-1)}$. 
 The cyclic code $C_{(n_{1},n_{2},q)}$ has parameters $[n,n_{1}+n_{2}-1,d_{(n_{1},n_{2},q)}]$, where $d_{(n_{1},n_{2},q)}=\mathrm{min}(n_{1},n_{2})$.
\end{theorem}

\begin{theorem}\label{minw0.1}
Assume that $q\in W_{0}$. Let $C_{(n_i,q)}^{(i,j)}$ denote the cyclic code over $\mathrm{GF}(q)$ with the generator polynomial 
$g_{(n_i,q)}^{(i,j)}(x)=\frac{x^{n}-1}{(x^{n_{i}}-1)\omega_{j}(x)}$, where $i={1,2}$, and $0 \leqslant j \leqslant 5.$
The cyclic code $C_{(n_i,q)}^{(i,j)}$ has parameters $[n,n_{i}+\frac{(n_{1}-1)(n_{2}-1)}{6},d_{(n_i,q)}^{(i,j)}]$, 
where $d_{(n_i,q)}^{(i,j)}\geq\lceil\sqrt{n_{i-(-1)^{i}}}\rceil$.\\
If $-1\in W_3,$ we have ${(d_{(n_i,q)}^{(i,j)}})^{2}- d_{(n_i,q)}^{(i,j)} + 1\geq n_{i-(-1)^{i}}.$
\end{theorem}

\begin{proof} Let $c(x)\in \mathrm{GF}(q)[x]/(x^{n}-1)$ be a codeword of Hamming weight $\omega$ in $C_{(n_i,q)}^{(i,j)}$. 
Take any $r\in W_{k}$ for $1\leq k\leq5$. Then $c(x^{r})$ is a codeword of Hamming weight $\omega$ in $C_{(n_i,q)}^{(i,(j-k)\ \mathrm{mod}\ 6)}$. 
It then follows that $d_{(n_i,q)}^{(i,j)}=d_{(n_i,q)}^{(i,(j-k)\ \mathrm{mod}\ 6)}.$ Therefore, we have 
$ d_{(n_i,q)}^{(i,0)}=d_{(n_i,q)}^{(i,1)}=d_{(n_i,q)}^{(i,2)}=d_{(n_i,q)}^{(i,3)}=d_{(n_i,q)}^{(i,4)}=d_{(n_i,q)}^{(i,5)}.$
Let $c(x)\in \mathrm{GF}(q)[x]/(x^{n}-1)$ be a codeword of minimum weight in $C_{(n_i,q)}^{(i,j)}$. 
Then $c(x^{r})$ is a codeword of same weight in $C_{(n_i,q)}^{(i,(j-k)\ \mathrm{mod}\ 6)}$. 
Further, for any $r\in W_k$, we have  $c(x)c(x^{r})$ is a codeword of $C_{i}$, where $C_{i}$ denote the cyclic code over $\mathrm{GF}(q)$ with the generator polynomial $g_i(x)=\frac{x^n-1}{x^{n_i}-1}$ and 
minimum distance $d_{i}=n_{i-(-1)^{i}}$. Hence, from Theorem \ref{notinw0}, we have ${(d_{(n_i,q)}^{(i,j)}})^{2}\geq d_{i}=n_{i-(-1)^{i}},$ and
${(d_{(n_i,q)}^{(i,j)}})^{2}- d_{(n_i,q)}^{(i,j)} + 1\geq n_{i-(-1)^{i}}$ if $-1\in W_3.$
\end{proof}
\begin{theorem}\label{minw0.2}
Assume that $q\in W_{0}$. Let $C_{(n_1,n_2,q)}^{(j)}$ denote the cyclic code over $\mathrm{GF}(q)$ 
with the generator polynomial $g_{(n_1,n_2,q)}^{(j)}(x)=\frac{(x^{n}-1)(x-1)}{(x^{n_{1}}-1)(x^{n_{2}}-1)\omega_{j}(x)}$, 
where $0 \leqslant j \leqslant 5.$
The cyclic code $C_{(n_1,n_2,q)}^{(j)}$ has parameters $[n,n_{1}+n_{2}-1+\frac{(n_{1}-1)(n_{2}-1)}{6},d_{(n_1,n_2,q)}^{(j)}]$, 
 where $d_{(n_1,n_2,q)}^{(j)}\geq\lceil\sqrt{\mathrm{min}(n_{1},n_{2})}\rceil$.\\
If $-1\in W_3,$ we have ${(d_{(n_1,n_2,q)}^{(j)}})^{2}- d_{(n_1,n_2,q)}^{(j)} + 1\geq \mathrm{min}(n_{1},n_2).$
\end{theorem}

\begin{proof} Let $c(x)\in \mathrm{GF}(q)[x]/(x^{n}-1)$ be a codeword of Hamming weight $\omega$ in $C_{(n_1,n_2,q)}^{(j)}$. 
Take any $r\in W_{k}$ for $1\leq k\leq5$. Then $c(x^{r})$ is a codeword of Hamming weight $\omega$ in $C_{(n_1,n_2,q)}^{((j-k)\ \mathrm{mod}\ 6)}$. 
It then follows that $d_{(n_1,n_2,q)}^{(j)}=d_{(n_1,n_2,q)}^{((j-k)\ \mathrm{mod}\ 6)}.$ Therefore, we have 
$ d_{(n_1,n_2,q)}^{(0)}=d_{(n_1,n_2,q)}^{(1)}=d_{(n_1,n_2,q)}^{(2)}=d_{(n_1,n_2,q)}^{(3)}=d_{(n_1,n_2,q)}^{(4)}=d_{(n_1,n_2,q)}^{(5)}.$
Let $c(x)\in \mathrm{GF}(q)[x]/(x^{n}-1)$ be a codeword of minimum weight in  $C_{(n_1,n_2,q)}^{(j)}$. 
Then $c(x^{r})$ is a codeword of same weight in $C_{(n_1,n_2,q)}^{((j-k)\ \mathrm{mod}\ 6)}$. 
Further, for any $r\in W_k$, we have $c(x)c(x^{r})$ is a codeword of $C_{(n_1,n_2,q)}$, where $C_{(n_1,n_2,q)}$ denote the cyclic code
over $\mathrm{GF}(q)$ with the generator polynomial $g(x)=\frac{(x^n-1)(x-1)}{(x^{n_1}-1)(x^{n_2}-1)}$ and 
minimum distance $d_{(n_1,n_2,q)}=\mathrm{min} (n_1, n_2).$ 
Hence, from Theorem \ref{notinw0.1}, we have ${(d_{(n_1,n_2,q)}^{(j)}})^{2}\geq d_{(n_1,n_2,q)}=\mathrm{min}(n_{1},n_2),$ and
${(d_{(n_1,n_2,q)}^{(j)}})^{2}- d_{(n_1,n_2,q)}^{(j)} + 1\geq \mathrm{min}(n_{1},n_2)$ if $-1\in W_3.$
\end{proof}
\begin{theorem}\label{minw0.3}
Assume that $q\in W_{0}$. Let $C_{(n_i,q)}^{(i,j,h)}$ denote the cyclic code over $\mathrm{GF}(q)$ with the generator 
polynomial $g_{(n_i,q)}^{(i,j,h)}(x)=\frac{x^{n}-1}{(x^{n_{i}}-1)\omega_{j}(x)\omega_{h}(x)}$, where $i={1,2}$ and \\
$(j,h)\in \{(0,1),(0,2),(1,2),(1,3),(2,3),(2,4),(3,4),(3,5),(4,0),(4,5),(5,0),(5,1)\}.$\\
The cyclic code $C_{(n_i,q)}^{(i,j,h)}$ has parameters $[n,n_{i}+\frac{(n_{1}-1)(n_{2}-1)}{3},d_{(n_i,q)}^{(i,j,h)}]$, 
where $d_{(n_i,q)}^{(i,j,h)}\geq\lceil\sqrt{n_{i-(-1)^{i}}}\rceil$. If $-1\in W_3,$ we have ${(d_{(n_i,q)}^{(i,j,h)}})^{2}- d_{(n_i,q)}^{(i,j,h)} + 1\geq n_{i-(-1)^{i}}$. 
\end{theorem}

\begin{proof} Let $j=0, h=1$ and $c(x)\in \mathrm{GF}(q)[x]/(x^{n}-1)$ be a codeword of Hamming weight $\omega$ in $C_{(n_i,q)}^{(i,0,1)}$. 
Take any $r\in W_{k}$ for $1\leq k\leq5$. Then $c(x^{r})$ is a codeword of Hamming weight $\omega$ in $C_{(n_i,q)}^{(i,(0-k)\ \mathrm{mod}\ 6,(1-k)\ \mathrm{mod}\ 6)}$.
It then follows that $d_{(n_i,q)}^{(i,0,1)}=d_{(n_i,q)}^{(i,(0-k)\ \mathrm{mod}\ 6,(1-k)\ \mathrm{mod}\ 6)}.$ Therefore, we have 
\begin{align}
 d_{(n_i,q)}^{(i,0,1)}=d_{(n_i,q)}^{(i,5,0)}=d_{(n_i,q)}^{(i,4,5)}=d_{(n_i,q)}^{(i,3,4)}=d_{(n_i,q)}^{(i,2,3)}=d_{(n_i,q)}^{(i,1,2)}.\label{2.1}
\end{align}
Let $j=0, h=2$ and $c(x)\in \mathrm{GF}(q)[x]/(x^{n}-1)$ be a codeword of Hamming weight $\omega$ in $C_{(n_i,q)}^{(i,0,2)}$. 
Take any $r\in W_{k}$ for $1\leq k\leq5$. Then $c(x^{r})$ is a codeword of Hamming weight $\omega$ in $C_{(n_i,q)}^{(i,(0-k)\ \mathrm{mod}\ 6,(2-k)\ \mathrm{mod}\ 6)}$.
It then follows that $d_{(n_i,q)}^{(i,0)}=d_{(n_i,q)}^{(i,(0-k)\ \mathrm{mod}\ 6,(2-k)\ \mathrm{mod}\ 6)}.$ Therefore, we have 
\begin{align}
 d_{(n_i,q)}^{(i,0,2)}=d_{(n_i,q)}^{(i,5,1)}=d_{(n_i,q)}^{(i,4,0)}=d_{(n_i,q)}^{(i,3,5)}=d_{(n_i,q)}^{(i,2,4)}=d_{(n_i,q)}^{(i,1,3)}.\label{2.2}
\end{align}
Further, from (\ref {2.1}) and (\ref {2.2})  for any $r\in W_{3}$, we have that $c(x)c(x^{r})$ is a codeword of $C_{i}$. where $C_{i}$ denote the cyclic code over $\mathrm{GF}(q)$ with the generator polynomial $g_i(x)=\frac{x^n-1}{x^{n_i}-1}$ and 
minimum distance $d_{i}=n_{i-(-1)^{i}}.$
Hence, from Theorem \ref{notinw0}, we have $({d_{n_i,q}^{i,j,h}})^{2}\geq d_{i}=n_{i-(-1)^{i}},$ i.e., $d_{(n_i,q)}^{(i,j,h)}\geq\lceil\sqrt{n_{i-(-1)^{i}}}\rceil$
and ${(d_{(n_i,q)}^{(i,j,h)}})^{2}- d_{(n_i,q)}^{(i,j,h)} + 1\geq n_{i-(-1)^{i}}$ if $-1\in W_3.$ 
\end{proof} 
\begin{theorem}\label{minw0.4}
Assume that $q\in W_{0}$. Let $C_{(n_1,n_2,q)}^{(j,h)}$ denote the cyclic code over $\mathrm{GF}(q)$ with the generator 
polynomial $g_{(n_1,n_2,q)}^{(j,h)}(x)=\frac{(x^{n}-1)(x-1)}{(x^{n_{1}}-1)(x^{n_{2}}-1)\omega_{j}(x)\omega_{h}(x)}$, where\\
$(j,h)\in \{(0,1),(0,2),(1,2),(1,3),(2,3),(2,4),(3,4),(3,5),(4,0),(4,5),(5,0),(5,1)\}.$\\
The cyclic code $C_{(n_1,n_2,q)}^{(j,h)}$ has parameters $[n,n_{1}+n_2-1+\frac{(n_{1}-1)(n_{2}-1)}{3},d_{(n_1,n_2,q)}^{(j,h)}]$, where $d_{(n_1,n_2,q)}^{(j,h)}\geq\lceil\sqrt{\mathrm{min}(n_{1},n_2)}\rceil$.
If $-1\in W_3,$ we have ${(d_{(n_1,n_2,q)}^{(j,h)}})^{2}- d_{(n_1,n_2,q)}^{(j,h)} + 1\geq \mathrm{min}(n_1,n_2)$. 
\end{theorem}

\begin{proof} Let $j=0, h=1$ and $c(x)\in \mathrm{GF}(q)[x]/(x^{n}-1)$ be a codeword of Hamming weight $\omega$ in $C_{(n_1,n_2,q)}^{(0,1)}$. 
Take any $r\in W_{k}$ for $1\leq k\leq5$, then $c(x^{r})$ is a codeword of Hamming weight $\omega$ in $C_{(n_1,n_2,q)}^{(i,(0-k)\ \mathrm{mod}\ 6,(1-k)\ \mathrm{mod}\ 6)}$.
It then follows that $d_{(n_1,n_2,q)}^{(0,1)}=d_{(n_1,n_2,q)}^{((0-k)\ \mathrm{mod}\ 6,((1-k)\ \mathrm{mod}\ 6)}.$ Therefore, we have
\begin{align}
 d_{(n_1,n_2,q)}^{(0,1)}=d_{(n_1,n_2,q)}^{(5,3)}=d_{(n_1,n_2,q)}^{(4,5)}=d_{(n_1,n_2,q)}^{(3,4)}=d_{(n_1,n_2,q)}^{(2,3)}=d_{(n_1,n_2,q)}^{(1,2)}.\label{2.3}
\end{align}
Let $j=0, h=2$ and $c(x)\in \mathrm{GF}(q)[x]/(x^{n}-1)$ be a codeword of Hamming weight $\omega$ in $C_{(n_1,n_2,q)}^{(0,2)}$. 
Take any $r\in W_{k}$ for $1\leq k\leq5$, then $c(x^{r})$ is a codeword of Hamming weight $\omega$ in $C_{(n_1,n_2,q)}^{((0-k)\ \mathrm{mod}\ 6,(2-k)\ \mathrm{mod}\ 6)}$.
It then follows that $d_{(n_1,n_2,q)}^{(0,2)}=d_{(n_1,n_2,q)}^{((0-k)\ \mathrm{mod}\ 6,(2-k)\ \mathrm{mod}\ 6)}.$ Therefore, we have
\begin{align}
 d_{(n_1,n_2,q)}^{(0,2)}=d_{(n_1,n_2,q)}^{(5,1)}=d_{(n_1,n_2,q)}^{(4,0)}=d_{(n_1,n_2,q)}^{(3,5)}=d_{(n_1,n_2,q)}^{(2,4)}=d_{(n_1,n_2,q)}^{(1,3)}.\label{2.4}
\end{align}
Further, from (\ref {2.3}) and (\ref {2.4})  for any $r\in W_{3}$, we have that $c(x)c(x^{r})$ is a codeword of $C_{i}$.
Hence, from Theorem \ref{notinw0.1}, we have ${(d_{(n_1,n_2,q)}^{(j,h)})}^{2}\geq d_{(n_1,n_2,q)}=\mathrm{min}(n_{1},n_2)$ i.e., $d_{(n_1,n_2,q)}^{(j,h)}\geq\lceil\sqrt{\mathrm{min}(n_1,n_2)}\rceil$,\\
and ${(d_{(n_1,n_2,q)}^{(j,h)}})^{2}- d_{(n_1,n_2,q)}^{(j,h)} + 1\geq \mathrm{min}(n_1,n_2)$ if $-1\in W_3.$ 
\end{proof}
\begin{theorem}\label{minw0.5}
Assume that $q\in W_{0}$. Let $C_{(n_i,q)}^{(i,j,h,l)}$ denote the cyclic code over $\mathrm{GF}(q)$ 
with the generator polynomial $g_{(n_i,q)}^{(i,j,h,l)}(x)=\frac{x^{n}-1}{(x^{n_{i}}-1)\omega_{j}(x)\omega_{h}(x)\omega_{l}(x)}$, where $i={1,2}$ and\\
$(j,h,l)\in \{(0,4,2),(0,4,5),(1,5,0),(2,0,1),(3,1,2),(4,2,3),(5,3,1),(5,3,4)\}$.\\  
The cyclic code $C_{(n_i,q)}^{(i,j,h,l)}$ has parameters $[n,n_{i}+\frac{(n_{1}-1)(n_{2}-1)}{2},d_{(n_i,q)}^{(i,j,h,l)}]$, 
where $d_{(n_i,q)}^{(i,j,h,l)}\geq\lceil\sqrt{n_{i-(-1)^{i}}}\rceil$.\\
If $-1\in W_3,$ we have ${(d_{(n_1,n_2,q)}^{(j,h,l)}})^{2}- d_{(n_1,n-2,q)}^{(j,h.l)} + 1\geq {n_{i-(-1)^{i}}}.$ 
\end{theorem}
\begin{proof} Let $j=0, h=4, l=2$ and $c(x)\in \mathrm{GF}(q)[x]/(x^{n}-1)$ be a codeword of Hamming weight $\omega$ in $C_{(n_i,q)}^{(i,0,4,2)}$. 
Take any $r\in W_{k}$ for $1\leq k\leq5$. 
The cyclic code $c(x^{r})$ is a codeword of Hamming weight $\omega$ in $C_{(n_i,q)}^{(i,(0-k)\ \mathrm{mod}\ 6,(4-k)\ \mathrm{mod}\ 6,(2-k)\ \mathrm{mod}\ 6)}$. 
It then follows that\\  $d_{(n_i,q)}^{(i,0,4,2)}=d_{(n_i,q)}^{(i,(0-k)\ \mathrm{mod}\ 6,(4-k)\ \mathrm{mod}\ 6,(2-k)\ \mathrm{mod}\ 6)}.$ Therefore, we have 
\begin{align}
 d_{(n_i,q)}^{(i,0,4,2)}=d_{(n_i,q)}^{(i,5,3,1)}=d_{(n_i,q)}^{(i,4,2,0)}=d_{(n_i,q)}^{(i,3,1,5)}=d_{(n_i,q)}^{(i,2,0,4)}=d_{(n_i,q)}^{(i,1,5,3)}.\label{3.1}
\end{align}
Let $j=0, h=4, l=5$ and $c(x)\in \mathrm{GF}(q)[x]/(x^{n}-1)$ be a codeword of Hamming weight $\omega$ in $C_{(n_i,q)}^{(i,0,4,5)}$. 
Take any $r\in W_{k}$ for $1\leq k\leq5$. Then $c(x^{r})$ is a codeword of Hamming weight $\omega$ in $C_{(n_i,q)}^{(i,(0-k)\ \mathrm{mod}\ 6,(4-k)\ \mathrm{mod}\ 6,(5-k)\ \mathrm{mod}\ 6)}$.
It then follows that\\ $d_{(n_i,q)}^{(i,0,4,5)}=d_{(n_i,q)}^{(i,(0-k)\ \mathrm{mod}\ 6,(4-k)\ \mathrm{mod}\ 6,(5-k)\ \mathrm{mod}\ 6)}.$ Therefore, we have 
\begin{align}
 d_{(n_i,q)}^{(i,0,4,5)}=d_{(n_i,q)}^{(i,5,3,4)}=d_{(n_i,q)}^{(i,4,2,3)}=d_{(n_i,q)}^{(i,3,1,2)}=d_{(n_i,q)}^{(i,2,0,1)}=d_{(n_i,q)}^{(i,1,5,0)}.\label{3.2}
\end{align}
Further, from (\ref {3.1}) and (\ref {3.2})  for any $r\in W_{3}$, we have that $c(x)c(x^{r})$ is a codeword of $C_{i}$, 
where $C_{i}$ and $d_i$ be defined as in Theorem \ref{notinw0}.
Hence, from Theorem \ref{notinw0}, we have $({d_{(n_i,q)}^{(i,j,h,l)}})^{2}\geq d_{i}=n_{i-(-1)^{i}},$ i.e., $d_{(n_i,q)}^{(i,j,h,l)}\geq\lceil\sqrt{n_{i-(-1)^{i}}}\rceil$ and
${(d_{(n_i,q)}^{(i,j,h,l)}})^{2}- d_{(n_i,q)}^{(i,j,h,l)} + 1\geq n_{i-(-1)^{i}}$, if $-1\in W_3$. 
\end{proof}
\begin{theorem}\label{minw0.6}
Assume that $q\in W_{0}$. Let $C_{(n_1,n_2,q)}^{(j,h,l)}$ denote the cyclic code over $\mathrm{GF}(q)$ with the generator polynomial
$g_{(n_1,n_2,q)}^{(j,h,l)}(x)=\frac{(x^{n}-1)(x-1)}{(x^{n_{1}}-1)(x^{n_{2}}-1)\omega_{j}(x)\omega_{h}(x)\omega_{l}(x)}$, where\\
$(j,h,l)\in \{(0,4,2),(0,4,5),(1,5,0),(2,0,1),(3,1,2),(4,2,3),(5,3,1),(5,3,4)\}$. 
The cyclic code $C_{(n_1,n_2,q)}^{(j,h,l)}$ has parameters $[n,n_{1}+n_{2}-1+\frac{(n_{1}-1)(n_{2}-1)}{2},d_{(n_1,n_2,q)}^{(j,h,l)}]$,
where $d_{(n_1,n_2,q)}^{(j,h,l)}\geq\lceil\sqrt{min(n_1,n_2)}\rceil$.\\
If $-1\in W_3,$ we have ${(d_{(n_1,n_2,q)}^{(j,h,l)}})^{2}- d_{(n_1,n-2,q)}^{(j,h.l)} + 1\geq \mathrm{min}(n_1,n_2).$\\

\end{theorem}

\begin{proof} Let $j=0,h=4,l=2$ and $c(x)\in \mathrm{GF}(q)[x]/(x^{n}-1)$ be a codeword of Hamming weight $\omega$ in $C_{(n_1,n_2,q)}^{(0,4,2)}$.
Take any $r\in W_{k}$ for $1\leq k\leq5$. Then $c(x^{r})$ is a codeword of Hamming weight $\omega$ in 
$C_{(n_1,n_2,q)}^{((0-k)\ \mathrm{mod}\ 6,(4-k)\ \mathrm{mod}\ 6,(2-k)\ \mathrm{mod}\ 6)}$.
It then follows that\\ $d_{(n_1,n_2,q)}^{(i,0,4,2)}=d_{(n_1,n_2,q)}^{(i,(0-k)\ \mathrm{mod}\ 6,(4-k)\ \mathrm{mod}\ 6,(2-k)\ \mathrm{mod}\ 6)}.$ Therefore, we have 
\begin{align}
 d_{(n_1,n_2,q)}^{(0,4,2)}=d_{(n_1,n_2,q)}^{(5,3,1)}=d_{(n_1,n_2,q)}^{(4,2,0)}=d_{(n_1,n_2,q)}^{(3,1,5)}=d_{(n_1,n_2,q)}^{(2,0,4)}=d_{(n_1,n_2,q)}^{(1,5,3)}.\label{3.3}
\end{align}
 Let $j=0,h=4,l=5$ and $c(x)\in \mathrm{GF}(q)[x]/(x^{n}-1)$ be a codeword of Hamming weight $\omega$ in $C_{(n_1,n_2,q)}^{(0,4,5)}$.Take any $r\in W_{k}$ for $1\leq k\leq5$.
 Then $c(x^{r})$ is a codeword of Hamming weight $\omega$ in $C_{(n_1,n_2,q)}^{((0-k)\ \mathrm{mod}\ 6,(4-k)\ \mathrm{mod}\ 6,(5-k)\ \mathrm{mod}\ 6)}$.
 It then follows that\\ $d_{(n_1,n_2,q)}^{(i,0,4,5)}=d_{(n_1,n_2,q)}^{(i,(0-k)\ \mathrm{mod}\ 6,(4-k)\ \mathrm{mod}\ 6,(5-k)\ \mathrm{mod}\ 6)}.$ Therefore, we have  
\begin{align}
 d_{(n_1,n_2,q)}^{(0,4,5)}=d_{(n_1,n_2,q)}^{(5,3,4)}=d_{(n_1,n_2,q)}^{(4,2,3)}=d_{(n_1,n_2,q)}^{(3,1,2)}=d_{(n_1,n_2,q)}^{(2,0,1)}=d_{(n_1,n_2,q)}^{(1,5,3)}.\label{3.4}
\end{align}
Further, from (\ref {3.3}) and (\ref {3.4})  for any $r\in W_{3}$, we have that $c(x)c(x^{r})$ is a codeword of $C_{(n_1,n_2,q)}$, 
where $C_{(n_1,n_2,q)}$ and $d_{(n_1,n_2,q)}$ be defined as in Theorem \ref{notinw0.1}.
Hence, from Theorem \ref{notinw0.1}, we have ${(d_{n_1,n_2,q}^{j,h,l})}^{2}\geq d_{(n_1,n_2,q)}=\mathrm{min}(n_1,n_2),$ i.e., $d_{(n_1,n_2,q)}^{(j,h,l)}\geq\lceil\sqrt{min(n_1,n_2)}\rceil$ and 
${(d_{(n_1,n_2,q)}^{(j,h,l)}})^{2}- d_{(n_1,n-2,q)}^{(j,h.l)} + 1\geq \mathrm{min}(n_1,n_2)$ if $-1\in W_3.$
\end{proof}
\begin{example}\label{ex1}
 Let $(p,m,n_1,n_2)=(2,1,13,19).$ Then $q=2,\ n=247$ and $C_s$ is a $[247,109]$ cyclic code over $\mathrm{GF}(q)$ with generator polynomial
$g(x)=\frac{x^{247}-1}{(x-1)\omega_0(x)\omega_1(x)\omega_2(x)}=x^{138}+x^{137}+x^{136}+x^{134}+x^{130}+x^{129}+x^{128}+x^{124}
+x^{121}+x^{120}+x^{111}+x^{107}+x^{106}+x^{105}+x^{104}+x^{102}+x^{98}+x^{96}+x^{94}+x^{93}+x^{92}+x^{88}+x^{87}+x^{86}+x^{85}
+x^{83}+x^{82}+x^{75}+x^{70}+x^{67}+x^{63}+x^{62}+x^{61}+x^{60}+x^{57}+x^{55}+x^{53}+x^{52}+x^{51}+x^{50}+x^{47}+x^{46}+x^{45}+x^{43}+x^{42}
+x^{40}+x^{39}+x^{38}+x^{37}+x^{36}+x^{35}+x^{34}+x^{33}+x^{31}+x^{30}+x^{28}+x^{27}+x^{26}+x^{25}+x^{22}+x^{21}+x^{16}+x^{13}
+x^{12}+x^{10}+x^9+x^8+x^6+x^3+x^2+1.$
We did some computation and our computation shows that upper bound on the minimum distance for this binary code is $48$.
 \end{example}
 \begin{example}\label{ex2}
 Let $(p,m,n_1,n_2)=(3,1,7,13).$ Then $q=2,\ n=91$ and $C_s$ is a $[91,19,7]$ cyclic code over $\mathrm{GF}(q)$ with generator polynomial
$g(x)=\frac{(x^{91}-1)(x-1)}{(x^7-1)(x^{13}-1)}=x^{72}+x^{71}+x^{65}+x^{64}+x^{59}+x^{57}+x^{52}+x^{50}+x^{46}+x^{43}+x^{39}+x^{36}+x^{33}+x^{29}+x^{26}+x^{22}+x^{20}
+x^{15}+x^{13}+x^{8}+x^{7}+x+1.$
 This is a bad cyclic code due to its poor minimum distance. The code in this case is bad because $q\notin W_0.$
 \end{example}
 \begin{example}\label{ex3}
 Let $(p,m,n_1,n_2)=(3,1,7,19).$ Then $q=3,\ n=133$ and $C_s$ is a $[133,61]$ cyclic code over $\mathrm{GF}(q)$ with generator polynomial
$g(x)=\frac{x^{133}-1}{(x^7-1)\omega_3(x)\omega_4(x)\omega_2(x)}= x^{72}+2x^{66}+x^{65}+2x^{64}+2x^{63}+2x^{61}+x^{60}+2x^{58}+2x^{57}+2x^{56}+x^{55}
+x^{54}+x^{53}+2x^{52}+x^{50}+2x^{49}+x^{48}+2x^{47}+2x^{46}+2x^{45}+x^{44}+2x^{41}+x^{40}+x^{37}+2x^{36}+x^{35}+x^{32}+2x^{31}
+x^{28}+2x^{27}+2x^{26}+2x^{25}+x^{24}+2x^{23}+x^{22}+2x^{20}+x^{19}+x^{18}+x^{17}+2x^{16}+2x^{15}+2x^{14}+x^{12}+2x^{11}+2x^{9}
+2x^{8}+x^{7}+2x^{6}+1.$ 
We did some computation and our computation shows that upper bound on the minimum distance for this ternary code is $35$. 
From Theorem \ref{minw0.5}, we have lower bound on the minimum distance for this ternary code is $5.$
 \end{example}
 \begin{conclusion*}
  WGCS were used to construct cyclic codes in \cite{ding12} and \cite{white04}. The idea of constructing cyclic codes with two-prime WGCS-II of order 6 could be viewed as an
  extension of above these two papers. 
  The cyclic codes employed in this paper depend on $n_1, n_2 $ and $q.$ When $q\in W_0$, we get a good code. We expect that the codes in Examples \ref{ex1} and \ref{ex3} give good codes.
  When $q\notin W_0$, we get a bad code, for example, we get a bad  code in Example \ref{ex2}.
  Finally, we expect that cyclic codes described in this paper can be employed to construct the good cyclic codes of large length.
  \end{conclusion*}
   \bibliographystyle{plain}
   \bibliography{ref}
 
\end{document}